\documentclass[10pt]{article}
\usepackage[left=0.8in,right=0.8in,top=1in,bottom=1in]{geometry}
\usepackage{amsmath,amssymb,amsfonts,amsthm}
\usepackage{simplebnf}
\usepackage{syntax}
\usepackage{authblk}
\usepackage{graphicx}
\usepackage{hyperref}
\usepackage{bussproofs}
\usepackage{biblatex}
\usepackage{bussproofs}
\usepackage{tikz-cd}
\theoremstyle{definition}

\newtheorem*{cor}{Corollary}

\newtheorem*{defi}{Definition}

\newtheorem*{note}{Note}
\newtheorem*{lemma}{Lemma}
\newtheorem*{notation}{Notation}
\newtheorem*{prop}{Proposition}

\newtheorem*{thm}{Theorem}

\newtheorem{nthm}{Theorem}[section]
\newtheorem{nlemma}[nthm]{Lemma}

\newcommand{\Tau}{\mathrm{T}}
\newenvironment{bprooftree}
  {\leavevmode\hbox\bgroup}
  {\DisplayProof\egroup}

\title{Hypergraph rewriting and Causal structure of $\lambda-$calculus}
\author{Utkarsh Bajaj}
\affil{University of Waterloo}
\date{\today}

\begin{document}

\maketitle

\begin{abstract}
\noindent In this paper, we first study hypergraph rewriting in categorical terms in an attempt to define the notion of events and develop foundations of causality in graph rewriting. We introduce novel concepts within the framework of double-pushout rewriting in adhesive categories. Secondly, we will study the notion of events in $\lambda-$calculus, wherein we construct an algorithm to determine causal relations between events following the evaluation of a $\lambda-$expression satisfying certain conditions. Lastly, we attempt to extend this definition to arbitrary $\lambda-$expressions. 
\end{abstract}

\section{Introduction}
In computational systems, a fundamental question arises regarding the relationship between two events: specifically, how one event may influence or cause the other. The concept of causality in computer science was initially explored by Glynn Winskel \cite{EventStructures}, who introduced an abstract framework for characterizing events and the causal relationships between them. Since that time, significant advancements have been made in understanding the structures that underlie causality. The significance of causality is also evident in dynamical systems, as highlighted in Einstein's general theory of relativity, where the causal structure of a Lorentzian manifold uniquely determines the geometry of spacetime, modulo a scaling factor. Moreover, areas such as causal set theory in theoretical physics offer additional insights into this complex subject. Despite these developments, explicit descriptions of events and their causal relationships within specific computational systems have not been as extensively studied. In this paper, we begin by examining the notion of events and the causal relations that arise between successive events in the context of hypergraph rewriting, framed within categorical terms. For those seeking further exploration of hypergraph rewriting, the Wolfram Physics Project \cite{WPP} serves as an excellent resource. Notably, several tools related to hypergraph rewriting are available in the Wolfram Language, as demonstrated in \cite{MultiwayCode}.
\newline
\\
\noindent We start with a definition of a hypergraph.  A directed hypergraph is a finite set of vertices $V$ and edges $E$ (which can be infinite), where an edge can connect a finite number of vertices i.e. $E \subseteq \bigcup_{n \in \mathbb{N}} V^{n}$. A more helpful notion of hypergraphs comes from considering the edges as \textit{labelled}, which gives way to the following definition:
\begin{defi}[Directed Multihypergraph]
     A labelled directed multihypergraph is a finite set of vertices $V$ and a (not-necessarily finite) set of edges $E$ with a map $f : E \rightarrow  \bigcup_{n \in \mathbb{N}} V^{n}$.
\end{defi}
\noindent Note that multihypergraphs are included in the above definition as $f$ does not have to be injective. Also, we include the empty graph $\emptyset$ as a directed multihypergraph. From here on, we will use graph, hypergraph, and directed (multi)hypergraph interchangeably. Figure \ref{hypergraphsexamples} show some visual representations of hypergraphs (examples taken from \cite{class}). 
\begin{figure}
    \centering
    \includegraphics[scale = 0.28]{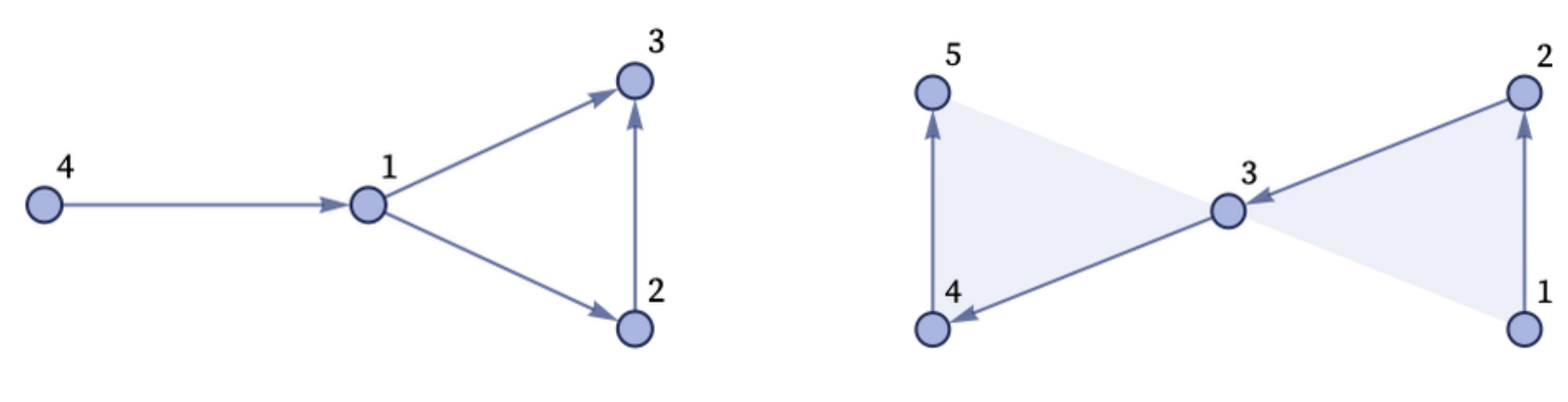}
    \caption{The edge lists are 
    \{(1,2),(2,3),(1,3)\} and  \{(1,2,3),(3,4,5)\} respectively} 
    \label{hypergraphsexamples}
\end{figure}
\noindent
Here, the edges connecting more than 2 vertices are shaded. There are other definitions of hypergraphs in the literature \cite{Gallo}, however, we will use the one presented above. In an undirected hypergraph, $E \subseteq \mathcal{P}(V) \setminus \emptyset$. A hypergraph rewriting system can be specified by an initial hypergraph and a collection of update rules. 

\begin{defi}[Update rule]
    An \textit{``update rule"} is a rewrite rule of the form $H_{1} \rightarrow H_{2}$ where $H_{1}$, $H_{2}$ are hypergraphs.
\end{defi}

\noindent
For instance, figure \ref{fig:2} is an example of an update rule.
\begin{figure}
    \centering
    \includegraphics[scale = 0.3]{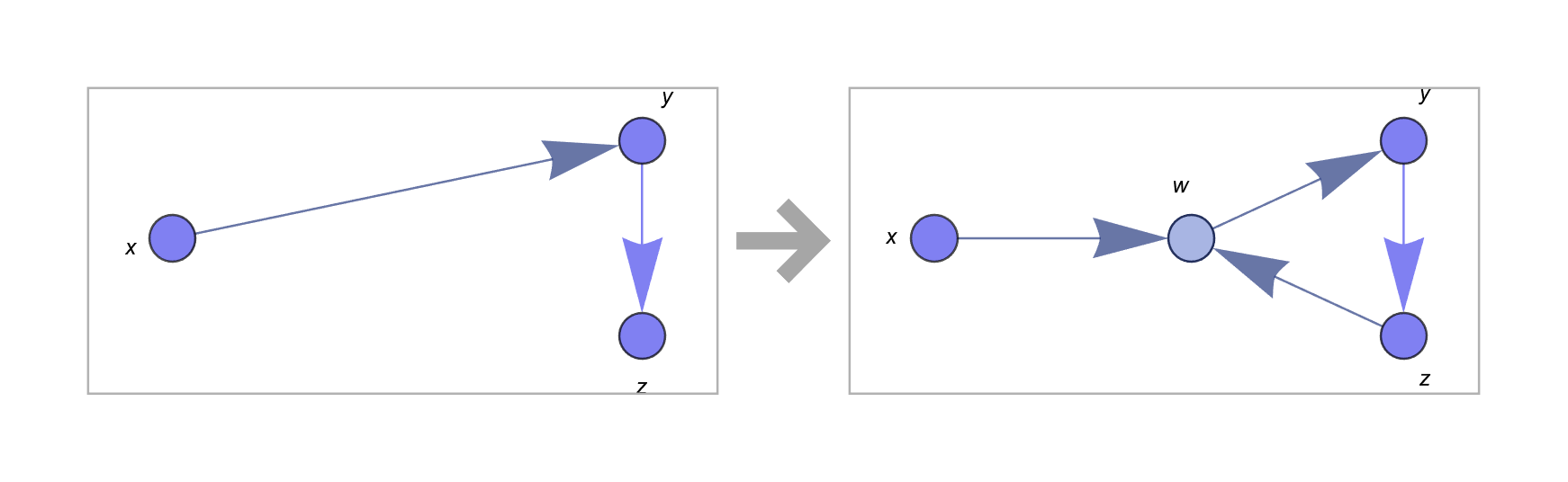}
    \caption{Example taken from \cite{class}. Here, the vertex which is shaded light blue on the right hand side graph is newly created} 
    \label{fig:2}
\end{figure}
\noindent
An update rule $H_{1} \rightarrow H_{2}$ is applied to a hypergraph $H$ by finding a subgraph of $H$ isomorphic to $H_{1}$, removing the vertices and edges deleted by the rule, and then gluing in $H_{2}$ along the vertices and edges preserved by the rule. Figure (\ref{fig:3}) illustrates how the rule in fig.(\ref{fig:2}) may be applied to some initial graph.
\begin{figure}
    \centering
    \includegraphics[scale = 0.4]{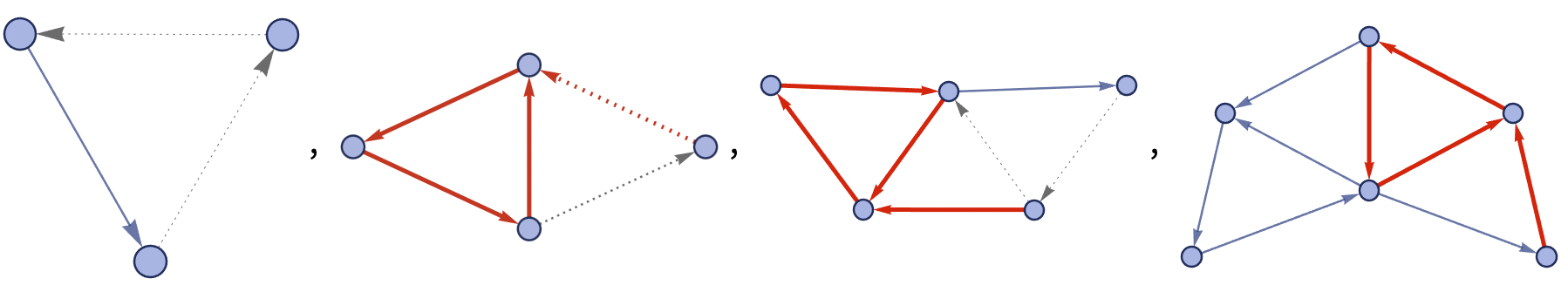}
    \caption{The dotted edges represent the part of the graph being removed. And the red edges represent the new edges added}
    \label{fig:3}
\end{figure}
\noindent We can describe operation of applying a rule in set theoretic terms. We first define some objects to make notation easier. To begin with, we can denote $\bigcup_{n \in \mathbb{N}} V^n$ as $V^{+}$ for any set $V$, the kleene plus operator on a finite set, which can be thought of as a covariant endofunctor $K^{+} : \mathbf{Set} \rightarrow \mathbf{Set}$, $V \mapsto V^{+}$. We can also define what a \textit{morphism} between 2 hypergraphs is:

\begin{defi}[Morphism] A morphism $h$ between hypergraphs $(V_{1}, E_{1}, f_{1})$ and $(V_{2}, E_{2},f_{2})$ consists of maps $h_V : V_{1} \rightarrow V_{2}$, $h_E : E_{1} \rightarrow  E_{2}$ such that the following diagram commutes:
\begin{center}
    \begin{tikzcd}
E_{1} \arrow[r, "f_{1}"] \arrow[d, "h_{E}"'] & K^{+}(V_{1}) \arrow[d, "K^{+}(h_V)"] \\
E_{2} \arrow[r, "f_2"']                      & K^{+}(V_{2})                
\end{tikzcd}
\end{center}
\end{defi}
\noindent 
This definition reduces the ordinary notion of graph morphisms, where an edge $e \in V \times V$. An update rule is of the form $L \rightarrow R$, where both $L = (V_{L}, E_{L}, f_{L})$ and $R = (V_{R}, E_{R}, f_{R})$ are hypergraphs. In our interpretation of a rewrite rule, the new vertices and edges are the ones with new names. To avoid a lot of unnecessary notation, we can formalize the notion of a rewrite rule:
\begin{defi}[Rewrite rule]
    A rewrite rule consists of graphs $I, L, R$ with monomorphisms (injective graph morphisms) $L \overset{l}{\hookleftarrow} I \overset{r}{\hookrightarrow} R$.
\end{defi}
\noindent
Here, the graph $I$ is called the interface graph, specifying the part of the left-hand side that is conserved by the rule. All the remaining vertices and edges are deleted. Then, after deletion, the graph $R$ is ``glued" along the graph $I$. For an update rule to be applied to some hypergraph $H = (V,E,f)$, we need a monomorphism $m : L \rightarrow H$ (called a \textit{matching}) which identifies the subgraph $m(L)$ to which the rule will be applied (the monomorphism condition is used because $L$ must be a subgraph of $H$).


\begin{defi}[Deletion of vertices]
    The graph formed by deleting a subset $S \subseteq V$ is $(V \setminus S, E', h')$ where $E' = \{e \in E : h(e) \in K^+(V \setminus S) \}$, and $h' = h \vert_{E'}$. This is denoted as $G \setminus S$.
\end{defi}
\noindent First, we delete the vertices $m_V(L \setminus I)$. So we get $H \setminus m_V(L \setminus I) = (V',E',f')$. In doing so, we have deleted a subset of edges deleted by the rule. We then delete the remaining edges $m_E(L \setminus I)$ to get the new edge set $\Tilde{E}$, and the new $\Tilde{f} = f'\vert_{\Tilde{E}}$. 


\begin{defi}[Cut graph]
    The graph $\Tilde{H}$ formed after deleting the vertices and edges, $(\Tilde{V}, \Tilde{E}, \Tilde{f})$ is called the \textit{cut graph}.
\end{defi}
\noindent We will show later that the cut graph construction obeys a universal property in the category of directed multi-hypergraphs. Note that by construction, $\Tilde{m} = m\vert_{I}$ is a morphism from $I$ to $\Tilde{H}$. Now, to get the output of the rewrite rule application, we can glue the cut graph $\Tilde{H}$ with $R$ along $I$. Formally, we construct the graph $(\Tilde{V} \sqcup V_{R} / \sim, \Tilde{E} \sqcup E_{R} / \sim, F)$ where $\Tilde{m}_{V}(v) \sim r_{V}(v)$ and $\Tilde{m}_{E}(e) \sim r_{V}(e)$, and $F : \Tilde{E} \sqcup E_{R} / \sim \rightarrow K^{+}(\Tilde{V} \sqcup V_{R} / \sim)$ takes $[\Tilde{e}] \mapsto [\Tilde{f}(\Tilde{e})]$ and $[e_{R}] \mapsto [f_{R}(e_{R})]$, where $\Tilde{e}$ is in $\Tilde{E}$ and $e_{R}$ is in $E_{R}$. Note that $K^{+}(\Tilde{V} \sqcup V_{R} / \sim) \cong K^{+}(\Tilde{V} \sqcup V_{R}) / \sim$ where $(v_{1},...,v_{n}) \sim (v'_{1},...,v'_{n})$ if $v_{i} \sim v'_{i}$ for all $1 \leq i \leq n$. It is easy to check that $F$ is well-defined. 
\newline \\
\noindent We will show in the next section that all of the above can be written elegantly in category-theoretic terms, which is called \textit{double-pushout rewriting} in the literature. An excellent explanation can be found in \cite{kissingerthesis}. 
\section{Categorical formulation of Hypergraph rewriting}
\noindent Let $\mathcal{H}$ denote the category of directed multihypergraphs with morphisms as defined in the previous section. If $f : H_{1} \rightarrow H_{2}, g : H_{2} \rightarrow H_{3}$ are morphisms, we can describe its composition, $g \circ f$, as the composition of the vertex and edge maps. This is indeed a morphism because $K^+$ is a functor i.e. $K^+(f \circ g) = K^+(f) \circ K^+(g)$. It is easy to see that the terminal object in this category is $(\{1\},\{e_1,e_2,..\}, f)$ where $f(e_n) = (1,1,..1)$, consisting of $n$ 1s. Also, the initial object is the empty graph $\emptyset$.
\begin{notation}
    Let $(a_1,...,a_n)$ and $(b_1,...,b_n)$ be two $n-$tuples. Then $(a_1,...,a_n) \times (b_1,...,b_n) := ((a_1, b_1),...,(a_n,b_n))$. This is not defined for two tuples which are not of the same length.
\end{notation}
\begin{prop} $\mathcal{H}$ has pullbacks.
\begin{center}
\begin{tikzcd}
J \arrow[rd, "\gamma" description, dashed] \arrow[rdd, "q_J"', bend right] \arrow[rrd, "p_J", bend left] &                                                             & \\ & H_1 \times_{\mathcal{H}} H_2 \arrow[r, "p"] \arrow[d, "q"'] & H_1 \arrow[d, "g"] \\& H_2 \arrow[r, "f"']  & H                 
\end{tikzcd}
\end{center}
    \begin{proof}
    Let $g \in \text{Mor}_{\mathcal{H}}(H_{1}, H)$, $f \in \text{Mor}_\mathcal{H}(H_{2}, H)$. Let $H_1 = (V_1, E_1, h_1)$ and $H_2 = (V_2, E_2, h_2)$. Then, construct $H_1 \times_{\mathcal{H}} H_2 = (V_F, E_F, H)$ where $V_F = \{(v_1,v_2) : g_V(v_1) = f_V(v_2)\}$, $E_F = \{(e_1, e_2) : e_1 \in E_1, e_2 \in E_2, g_E(e_1) = f_E(e_2) \}$, and $H : E_F \rightarrow K^{+}(V_F)$ is given by $(e_1,e_2) \rightarrow h_1(e_1) \times h_2(e_2)$ (since $g$ is a morphism and $g_E(e_1) = f_E(e_2) \implies h(g_E(e_1)) = K^+(g_V)(h_1(e_1)) = K^+(g_V)(h_2(e_2)) = h(g_E(e_2))$, $h_1(e_1)$ and $h_2(e_2)$ must have the same length as tuples). $p_V, p_E$ (and $q_V, q_E$) are the canonical projection maps. $p$ is a morphism because $h_1(p_E(e_1,e_2)) = h_1(e_1) = K^+(p_V)(h_1(e_1) \times h_2(e_2)) = K^+(p_V)(H(e_1,e_2))$. Similarly with $q$. The diagram above commutes obviously. Assume that there is some other $J = (V_J, E_J, j)$ with morphisms $p_J : J \rightarrow H_1, q_J : J \rightarrow H_2$ such that $g \circ p_J = f \circ q_J$. Then construct $\gamma : J \rightarrow H_1 \times_{\mathcal{H}} H_2$ where $\gamma_V(v) = (p_{J_V}(v), q_{J_V}(v)), \gamma_E(e) = (p_{J_E}(e), q_{J_E}(e))$. This is the unique map that makes the diagram commute. This is also a morphism because $H(\gamma_E(e)) = h_1(p_{J_E}(e)) \times h_2(q_{J_E}(e))$, and $h_1(p_{J_E}(e)) = K^+(p_{J_V})(j(e))$, $h_2(q_{J_E}(e)) = K^+(q_{J_V})(j(e))$. By definition, $K^+(\gamma_V)(v_1,...,v_n) = (\gamma_V(v_1),...,\gamma_V(v_n)) = (p_{J_V}(v_1),...,p_{J_V}(v_n)) \times (q_{J_V}(v_1),...,q_{J_V}(v_n))$, so $H(\gamma_E(e)) = K^+(\gamma_V)(j(e))$.
    \end{proof}
\end{prop}
\noindent Since this category has pullbacks and a terminal object, finite limits exist \cite{MacLane}. Also, note that $f : H_1 \rightarrow H_2$ is a monomorphism iff $f_V$ and $f_E$ are injective as set maps.

\begin{prop}
$\mathcal{H}$ has pushouts.
\end{prop}
\begin{center}
\begin{tikzcd}
H \arrow[r, "f"] \arrow[d, "g"']       & H_1 \arrow[d, "p"] \arrow[rdd, "p'", bend left]         &    \\
H_2 \arrow[r, "q"'] \arrow[rrd, "q'"', bend right] & H_1 \sqcup H_2 \arrow[rd, "\gamma" description, dashed] &    \\
                                                   &                                                         & H'
\end{tikzcd}
\end{center}

\begin{proof}
    Let $H_1 = (V_1, E_1, h_1)$ and $H_2 = (V_2,E_2, h_2)$, $H = (V,E,h)$, and $f : H \rightarrow H_1$ is a monomorphism. We claim that the pushout of these is $(V_1 \sqcup V_2 \slash \sim, E_1 \sqcup E_2 \slash \sim, H)$ where $f_V(v) \approx g_V(v)$ and $f_E(e) \approx g_E(e)$, and $\sim$ is the equivalence relation generated by $\approx$ (so the smallest equivalence relation containing $\approx$). 
    $H : E_1 \sqcup E_2 \slash \sim \rightarrow K^+(V_1 \sqcup V_2) \slash \sim \cong K^+(V_1 \sqcup V_2 \slash \sim)$, where $[e_1] \mapsto [h_1(e_1)]$, and $[e_2] \mapsto [h_2(e_2)]$. We now show that the map is well-defined. If $e \sim e'$ and $e \neq e'$ (WLOG assume $e \in E_1)$, then by definition, we must have that $e = e_1 \approx e_2 \approx ... \approx e_n = e'$ where $n \geq 2$. Thus, $e_1 = f_E(e), e_2 = g_E(e)$ for some $e \in E$, and $h_1(e_1) =h_1(f_E(e)) = K^+(f_V)(h(e))$, $h_2(e_2) =  h_2
(g_E(e)) = K^+(g_V)(h(e))$. If $h(e) = (v_1,...,v_m)$, then $h_1(e_1) = (f_V(v_1),...,f_V(v_m))$, and $h_2(e_2) = (g_V(v_1),...,g_V(v_m))$. These are in the same equivalence class. Thus, $H(e_1) \sim H(e_2)$. Continuing like this, we will get that $H(e_1) \tilde H(e_n)$, which means that the map is well-defined.  Let $p_V : V_1 \rightarrow V_1 \sqcup V_2 \slash \sim$, $v_1 \mapsto [v_1]$ (and similarly define $p_E, q_V, q_E$ over here). $p$ is a morphism because $H(p_E(e_1)) = [h_1(e_1)] = K^+(p_V)(h_1(e_1))$ by definition. Similarly, $q$ is a morphism. Also, it is easy to see that $p \circ f = q \circ g$. Assume that $p' : H_1 \rightarrow H'$ and $q' : H_2 \rightarrow H'$ are morphisms of hypergraphs such that $p' \circ f = q' \circ g$. Then, construct $\gamma : H_1 \sqcup H_2 \rightarrow H'$, where $\gamma_V : V_1 \sqcup V_2 \slash \sim \rightarrow V'$, as $[v] \mapsto p'_{V}(v)$ if $v \in V_1$ and $[v] \mapsto q'_{V}(v)$ if $v \in  V_2$ (define $\gamma_E$ similarly). We claim that this map is well-defined. If $[v] = [v']$ and $v \neq v'$ (WLOG $v \in V)$, then $v = v_1 \approx v_2 \approx ... \approx v_n = v'$ where $n \geq 2$. Thus, $v_1 = f_V(v), v_2 = g_V(v)$ for some $v \in V$. Then, $\gamma_V([v_1]) = p'_{V}(f_V(v)) = q'_{V}(g_V(v)) = \gamma_V([v_2])$. Continuing like this, we will get that $\gamma_V([v]) = \gamma_V([v'])$ (because the diagram commutes). This is a morphism because $h'(\gamma_E([e_1])) = h'(p'_{E}(e_1)) = K^+(p'_V)(h_1(e_1))$, and this is the same as $K^+(\gamma_V)(H([e_1]))$ by definition. Also, this must be the unique such morphism since $\gamma([v_1])$ (where $v_1 \in H_1$) $ = \gamma(p(v_1)) = p'(v_1)$. 
\end{proof}
\noindent The machinery of pullbacks and pushouts can provide an alternative formulation of hypergraph rewriting. We first define the notion of a pushout complement \cite{kissingerthesis}:
\begin{defi}[Pushout complement]
    Let $I \overset{l}{\rightarrow} L \overset{m}{\rightarrow} G$ are morphisms, then $G'$ is called a \textit{pushout complement} if there exist morphisms $I \overset{m'}{\rightarrow} G$ and $G \rightarrow G'$ such that the following square is a pushout square:
    \begin{center}
        \begin{tikzcd}
I \arrow[r, "l"] \arrow[d, "m'"']              & L \arrow[d, "m"] \\
G' \arrow[r]                                   & G                \\              
\end{tikzcd}
    \end{center}
   
\end{defi} 
\begin{defi}[No-dangling-edges condition] A matching $m : L \rightarrow G$ with rewrite rule $ (L \overset{l}{\hookleftarrow} I \overset{r}{\hookrightarrow} R )$ satifies the \textit{no-dangling-edges} condition if for any $v \in V_L \setminus l(V_I)$, all edges containing $m_V(v)$ are of the form $m_E(e)$ for some $e \in E_L$.
\end{defi}
\noindent This means that all the edges containing vertices that must be deleted are inside $L$ itself. We now can relate this to pushout complements. Note that if $m_V(v)$ is contained in $m_E(e)$ where $e\ \in E_L$, then $m_V(v) \in f(m_E(e)) = K^+(m_V)(f_L(e))$ (where $G = (V,E,f)$) i.e. $v \in f_L(e)$. Thus, $e \in  E_L \setminus l(E_I)$ because $v \in V_L \setminus l(V_I)$.
\begin{prop}
    A matching $m : L \rightarrow G$ and a rewrite rule $
    (L \overset{l}{\hookleftarrow} I \overset{r}{\hookrightarrow} R )$ satisfies the no-dangling-edges condition iff $I \overset{l}{\hookrightarrow} L \overset{m}{\rightarrow} G$ has a pushout complement. 
\end{prop}

\begin{proof}
    Assume that the no-dangling edges condition is satisfied. We will show that the cut graph $G_c$ is the pushout complement. Denote $I = (V_I, E_I, f_I), L =(V_L,E_L, f_L), G= (V,E,f)$.  Then, the cut graph has vertices $V \setminus  m_V(V_L  \setminus l(V_I)0$ and edges $E \setminus m_E(E_L \setminus l(E_I))$ (because there are no dangling edges, when we delete the vertices, the edges which are deleted as a consequence lie inside $m_E(E_L \setminus l(E_I))$). We have $m \circ l : I \rightarrow G_c$ as a morphism because the image of $m$ is in $G$. The pushout of arrows $l : I \rightarrow L$ and $m \circ l : I \rightarrow G_c$ has vertex set $(V_L \sqcup V \setminus  m_V(V_L  \setminus l(V_I)))  \slash \sim$ where $l_V(v_i) \approx m_V(l_V(v_i))$ for every $v_i \in V_I$. What's the equivalence relation generated by this relation? Since $l_V, m_V$ are injective, we have that the equivalence relation generated by $\approx$ satisfies $v \sim w$ and $v \neq w$ implies that $v = l_V(v_i)$, $w = m_V(l_V(v_i))$ for some unique $v_i \in V_I$. Construct a map $\phi_V : (V_L \sqcup V \setminus  m_V(V_L  \setminus l(V_I)) ) \slash \sim \rightarrow V$, taking $v \in V_L \mapsto m_V(v)$ and $v' \in V \setminus  m_V(V_L  \setminus l(V_I)) \mapsto v'$. This map is well-defined because $l_V(v_i)$ is sent to $m_V(l_V(v_i))$ by definition, and $l_V(v_i) \sim m_V(l_V(v_i))$. It is easy to see that this map is bijective. Similarly, the pushout has edge set $(E_L \sqcup E \setminus m_E(E_L \setminus l(E_I))) \slash \sim$ where $l_E(e_i) \approx m_E(l_E(e_i))$ for any $e_i \in E_I$. As is in the case of vertices, the map $\phi_E : (E_L \sqcup E \setminus m_E(E_L \setminus l(E_I))) \slash \sim \rightarrow E$ defined similarly is bijective. It is easy to see that $\phi$ is an isomorphism from the pushout to $G$, and so the cut graph is a pushout complement. If the no-dangling condition is not satisfied, then there will be an edge in $G$ that is not in the $E_L$ which will be deleted. Thus, it won't be in the cut graph. Also, since it isn't in $E_L$, it won't be in $L$. So, it won't be in the pushout, which means that the cut graph is not the pushout complement.
\end{proof}

\noindent In the context of hypergraphs (and graphs, in general), a fibre product can thought of an intersection.  $H_1$ and $H_2$ can be thought of as subgraphs of $H$, and their intersection in the subgraph is the fibre product. Similarly, a pushout of monomorphisms can be viewed as a union along a common subgraph (gluing). The monomorphisms $H \hookrightarrow H_1$ and $H \hookrightarrow H_2$ are usually inclusion maps, and $H_1 \sqcup  H_2$ is a union along the common subgraph $H$. Using the machinery of pushouts and fibre products (often called pullbacks), we can easily understand graph rewriting. Given a match $m : L \rightarrow G$ and a rewrite rule $\rho = L \overset{l}{\hookleftarrow} I \overset{r}{\hookrightarrow} R$ which satisfies the no-dangling-edge condition, we first compute the cut graph, which is the pushout complement $G' = G_c$. Then, we glue $R$ and $G'$ along $I$. We get a double-pushout diagram (both squares are pushout squares):
\begin{center}
    \begin{tikzcd}
L \arrow[d, "m"'] & I \arrow[l, "l"', hook'] \arrow[r, "r", hook] \arrow[d] & R \arrow[d] \\
G                 & G' \arrow[l] \arrow[r]                                  & H          
\end{tikzcd}
\end{center}
\noindent We call $H$ the graph \textit{production}, or the output graph. It is possible to do graph rewriting even if the no-dangling edge is not satisfied$-$we can construct the cut graph, and then glue the graph $R$. However, if we want to do rewriting in arbitrary categories, then double-pushout, or DPO, rewriting is the preferred formalism \cite{kissingerthesis}. In general, DPO rewriting can be done in $\textit{adhesive}$ categories, which are categories in which pushout complements are unique up to isomorphism. A concise introduction to adhesive categories is given in \cite{Adhesive}.
\begin{defi}[Adhesive category]
    A category $\mathbf{
    C}$ is called \textit{adhesive} if
    \begin{itemize}
        \item it has pushouts along monomorphisms
        \item it has pullbacks
        \item pushouts along monomorphisms are Van-Kampen squares
    \end{itemize}
\end{defi}
\noindent See \cite{Adhesive} for the definition of a Van-Kampen (VK) square. The uniqueness of pushout complements follows directly from the VK square condition. We can show that our category, $\mathcal{H}$, is adhesive. We first show that our category is extensive.
\begin{defi}[Extensive category]
    A category $\mathbf{C}$ is called $\textit{extensive}$ when 
    \begin{itemize}
        \item it has finite coproducts
        \item it has pullbacks along coproduct injections
        \item given a diagram where the bottom row is a coproduct, 
        \begin{center}
        \begin{tikzcd}
X \arrow[r, "m"] \arrow[d, "r"'] & Z \arrow[d, "h"'] & Y \arrow[l, "n"'] \arrow[d, "s"] \\
A \arrow[r, "i"']                & A \sqcup B        & B \arrow[l, "j"]                
\end{tikzcd}
\end{center}
        the 2 squares are pullbacks iff the top row is a coproduct
    \end{itemize}
\end{defi}
\begin{thm}$\mathcal{H}$ is an extensive category.
\end{thm}
\begin{proof}
    If the top row is a coproduct, we get that 
    $m$ and $n$ are inclusion maps and $Z = X \sqcup Y$, and $h = r \sqcup s$. Then the pullback of $i$ and $h$ is the graph with vertices $(v_{A}, w)$ such that $v_{A} = h(w)$. So, $w \in X$ and so $v_{A} = r(w)$, which means that the vertices are $(r(w), w)$ where $w \in X$. Repeating this with the edges, it is easy to see that the pullback is isomorphic to $X$ and by symmetry both squares are pullback squares. If both squares are pullbacks, we get the following diagram
\begin{center}
        \begin{tikzcd}
A \times_{A \sqcup B} Z \arrow[d] \arrow[r] & Z \arrow[d, "h" description] & B \times_{A \sqcup B} Z \arrow[l] \arrow[d] \\
A \arrow[r]                                 & A \sqcup B                   & B \arrow[l]                                
\end{tikzcd}
\end{center}
\noindent We have that $A \times_{A \sqcup B} Z$ has vertex set $\{(h(v_z),v_z) : h(v_z) \in A\}$ and $B \times_{A \sqcup B} Z$ has vertex set $\{(h(v_z), v_z) : h(v_z) \in B\}$. It is clear that the disjoint union of these 2 graphs is isomorphic to $Z$.
\end{proof}
\noindent We have the following lemma (\cite{Adhesive})
\begin{lemma}
    In an extensive category, pushouts along coproduct injections are VK squares.
\end{lemma}
\begin{cor} $\mathcal{H}$ is an adhesive category
\end{cor} 
\begin{proof}
    Since every monomorphism is a coproduct injection in $\mathcal{H}$ (take the coproduct with the empty graph), and $\mathcal{H}$ is adhesive, the result follows from the above lemma.
\end{proof}
 \section{Causality}
\begin{defi}[Event]
    An event is given by a rule $\rho = L \overset{l}{\hookleftarrow} I \overset{r}{\hookrightarrow} R$ and a matching $m : L \rightarrow G$ (which is a morphism).
\end{defi}
\noindent 
It is possible for two distinct events to be applied at different locations on a graph $G$, however resulting in isomorphic output graphs. Ideally, one would want those graphs to be distinguished, for which the double-pushout rewriting approach should not be used (because pushout complement and pushouts are described only up to isomorphism). To do so, a labelling function may be used that encodes the information about the event in the output graph. This will be seen in the next section in the context of determining causal relations in $\lambda-$calculus. For a preliminary discussion on causality, however, considering the isomorphism classes of $\mathcal{H}$ is sufficient. Moreover, the following can be generalised to any adhesive category, where a canonical labelling is not known.
\begin{note}
    From here on, a monomorphism $m : L \rightarrow G$  is a matching if it obeys the no-dangling-edge condition.
\end{note}
\noindent We can define a reduction relation on $\text{sk}(\mathcal{H})$ (the skeleton of $\mathcal{H}$).
\begin{defi}[$\rightarrow_{\beta}$]
    Let $H_1, H_2 \in \text{sk}(\mathcal{H})$. Then, $H_1 \rightarrow_{\beta} H_2$ if there exists a rewrite rule $\rho = L \overset{l}{\hookleftarrow} I \overset{r}{\hookrightarrow} R$ and a matching $m:L \rightarrow H_1$ such that the output graph is isomorphic to $H_2$. 
\end{defi}
\noindent This transition can be labelled $H_1 \overset{(\rho, m)}{\rightarrow_{\beta}} H_2$. Thus, $\mathcal{H}$ are the states of a labelled transition system (see \cite{Joyal} for the definition), where we label transitions with events.  We can also define a notion of morphism between rewrite rules.
\begin{defi}[Morphism of rewrite rules]
    Let $\rho = L \overset{l}{\hookleftarrow} I \overset{r}{\hookrightarrow} R$ and $\rho' =  L' \overset{l'}{\hookleftarrow} I' \overset{r'}{\hookrightarrow} R'$ be 2 rewrite rules. Then a morphism from $\rho$ to $\rho'$ consists of morphisms $f : L \rightarrow L'$, $h: I \rightarrow I'$, $g : R \rightarrow R'$ such that the following diagram commutes:
    \begin{center}
        \begin{tikzcd}
L \arrow[d, "f"'] & I \arrow[l, "l"', hook'] \arrow[r, "r", hook] \arrow[d, "h" description] & R \arrow[d, "g"] \\
L'                & I' \arrow[l, "l'", hook'] \arrow[r, "r'"', hook]                         & R'              
\end{tikzcd}
    \end{center}
\end{defi}
\begin{defi}[$\rightarrow_\mathcal{\rho}$]
    Let $\mathcal{R}$ be the category of rewrite rules, and $\rho = L \overset{l}{\hookleftarrow} I \overset{r}{\hookrightarrow} R \in \mathcal{R}$. Then $\rightarrow_\rho$ is a relation on $\text{sk}(\mathcal{H})$ defined by $H_1 \rightarrow_\rho H_2$ if there is some matching $m : L \rightarrow H_1$ and $H_2$ is isomorphic to the output graph. 
\end{defi}
\noindent We have that $\rightarrow_{\beta} = \bigcup_{\rho \in \mathcal{R}} \rightarrow_{\rho}$. It is easy to see that if $\rho$ is isomorphic to $\rho'$ in $\mathcal{R}$, then $\rightarrow_{\rho} = \rightarrow_{\rho'}$ on $\text{sk}(\mathcal{H})$. However, the converse is not true. Let's introduce some extra notation. If $G \in \text{sk}(\mathcal{H})$ has $n$ vertices, let $F(G)$ denote the unique graph in $\text{sk}(\mathcal{H})$ with $n$ vertices and no-edges. If $f : H_{1} \rightarrow H_{2}$ is a morphism, then let $F(f) : F(H_{1}) \rightarrow F(H_{2})$ be the canonical morphism with $F(f)_{V} = f_{V}$ and $F(f)_{E}$ be the empty map (from the empty set to the empty set). Also, given a graph $G$, there is the obvious inclusion morphism from $F(G)$ to $G$, which means that there is a canonical morphism $\Tilde{f} : F(H_{1}) \rightarrow H_{2}$ for every $f : H_{1} \rightarrow H_{2}$.
\begin{prop} Let $\rho_{1} = L_{1} \overset{l_{1}}{\hookleftarrow} I_{1} \overset{r_{1}}{\hookrightarrow} R_{1} \in \mathcal{R}$, $\rho
_{2} = L_{2} \overset{l_{2}}{\hookleftarrow} I_{2} \overset{r_{2}}{\hookrightarrow} R_{2} \in \mathcal{R}$ be 2 rewrite rules. Then, if the rewrite rules $\Tilde{\rho}_{1} = L_{1} \overset{\Tilde{l}_{1}}{\hookleftarrow} F(I_{1}) \overset{\Tilde{r}_{1}}{\hookrightarrow} R_{1} $ and $\Tilde{\rho}_{2} = L_{2} \overset{\Tilde{l}_{2}}{\hookleftarrow} F(I_{2}) \overset{\Tilde{r}_{2}}{\hookrightarrow} R_{2}$ are isomorphic, $\rightarrow_{\rho_{1}} = \rightarrow_{\rho_{2}}$.
\end{prop}
\begin{proof}
     ($\Leftarrow$)Let $(f,g,h) : \Tilde{\rho}_{1} \rightarrow \Tilde{\rho}_{2}$ be an isomorphism of rewrite rules. Let $G \in \text{sk}(\mathcal{H})$ such that $G \rightarrow_{\rho_{2}} H$ where $m : L_{2} \rightarrow G$ is the matching. We claim that $m \circ f: L_{1} \rightarrow G$ obeys the no-dangling-edges condition. Assume $v \in V_{L_{1}} \setminus l_{1}(V_{I_{1}})$. If $f(v) \in l_{2}(V_{I_{2}})$, then $f(v) = l_{2}(w)$ for some unique $w \in V_{I_{2}}$. Since the following diagram commutes, 
     \begin{center}
\begin{tikzcd}
L_{1} \arrow[d, "f"'] & F(I_{1}) \arrow[l, "\tilde{l}_{1}"', hook'] \arrow[d, "h"'] \\
L_{2}                 & F(I_{2}) \arrow[l, "\tilde{l}_{2}", hook']                 
\end{tikzcd}
     \end{center}
     $\Tilde{l}_{2} \circ h (h^{-1}(w)) = f(v) = f \circ \Tilde{l}_{1}(w)$. Since $f$ isomorphism, $v = \Tilde{l}_{1}(w)$ i.e. $v \in l_{1}(V_{I_{2}})$. Contradiction. So $f(v) \in V_{L_{2}} \setminus l_{2}(V_{I_{2}})$. And since $m$ obeys the no-dangling-edges condition, any edges containing $m(f(v))$ are of the form $m(e)$ for some $e \in  E_{L_{2}}$. Since $f$ is an isomorphism, $e = f(e')$ for some $e' \in E_{L_{1}}$. Thus, $m \circ f$ satisfies the no-dangling edges condition i.e. the pushout complement of $I_{1} \overset{l_{1}}{\rightarrow} L_{1} \overset{m \circ f}{\rightarrow} G$ exists, say $G_{1}$. It is easy to see, by the fact that pushout complement is equal to the cut graph, that removing the edges present in $I_{1}$ as a subgraph of $G_{1}$ will give the pushout complement of $F(I_{1}) \overset{\Tilde{l}_{1}}{\rightarrow} L_{1} \overset{m \circ f}{\rightarrow} G$, $\Tilde{G}_{1}$. In fact, the following is a pushout square:
     \begin{center}
         \begin{tikzcd}
\tilde{G}_{1} \arrow[d] & F(I_{1}) \arrow[l, "\tilde{l}_{1}"', hook'] \arrow[d] \\
G_{1}                   & I_{1} \arrow[l, hook']                               
\end{tikzcd}
     \end{center}
     \noindent This can also be derived from diagram chasing arguments, using the fact that the following diagram also commutes i.e. there is a monomorphism from $\Tilde{\rho}_{1}$ to $\rho_{1}$. 
     \begin{center}
         \begin{tikzcd}
L_{1} \arrow[d, "\text{id}"'] & F(I_{1}) \arrow[l, "\tilde{l}_{1}"', hook'] \arrow[r, "\tilde{r}_{1}"] \arrow[d, "\iota" description] & R_{1} \arrow[d, "\text{id}"] \\
L_{1}                         & I_{1} \arrow[l, "l_{1}", hook'] \arrow[r, "r_{1}"', hook]                                             & R_{1}                       
\end{tikzcd}
     \end{center}
     \noindent 
    Let $G_{1}$ be the pushout complement of $I_{1} \overset{l_{1}}{\rightarrow} L_{1} \overset{m \circ f}{\rightarrow} G$ and $H_{1}$ be the output graph when the first rewrite rule is applied i.e. it is the pushout of the pair $I_{1} \rightarrow R_{1}$ and $I_{1} \rightarrow G_{1}$. It can also be checked that the outer square is a pushout square (pushouts compose in $\mathcal{H}$). 
    \begin{center}
        \begin{tikzcd}
F(I_{1}) \arrow[r] \arrow[d] & I_{1} \arrow[r] \arrow[d] & R_{1} \arrow[d] \\
\tilde{G}_{1} \arrow[r]      & G_{1} \arrow[r]           & H_{1}          
\end{tikzcd}
    \end{center}
    Therefore, the output graph of $\tilde{\rho}_{1}$ and $\rho_{1}$ are isomorphic. Similarly, the output graph of $\tilde{\rho}_{2}$ and $\rho_{2}$ are isomorphic. Since $\tilde{\rho}_{1} \cong \tilde{\rho}_{2}$, their output graphs are isomorphic, which means that $H \cong H_{1}$, i.e. $G \rightarrow_{\rho_{1}} H$ in $\text{sk}(\mathcal{H})$. So, $\rightarrow_{\rho_{2}} \subseteq \rightarrow_{\rho_{1}}$. The other inclusion follows from symmetry.
\end{proof}
\noindent The backward direction is true for ``most" rewrite rules. Assume the number of edges in $L_{1}, L_{2}$ are finite. Since $L_{2} \rightarrow_{\rho_{2}} R_{2}$, we have $L_{2} \rightarrow_{\rho_{1}} R_{2}$ as well, implying that there is a monomorphism $m : L_{1} \rightarrow L_{2}$. By finiteness, this is an isomorphism, which means we get the following double-pushout diagram:
\begin{center}
    \begin{tikzcd}
L_{1} \arrow[d, "m"'] & I_{1} \arrow[l, "l_{1}"'] \arrow[d, "\text{id}"] \arrow[r] & R_{1} \arrow[d] \\
L_{2}                 & I_{1} \arrow[l, "m \circ l_{1}"] \arrow[r]                 & R_{1}          
\end{tikzcd}
\end{center}
Since the production must be $R_{2}$, $R_{1}$ is also isomorphic to $R_{1}$. For most rewrite rules, if we can show that there is no isomorphism from $L_{1}$ to $L_{2}$ that make the rewrite rules $\tilde{\rho}_{1}$, $\tilde{\rho}_{2}$ isomorphic (i.e. the diagram commutes), it is possible to create a matching for $L_{1}$ that isn't a matching for $L_{2}$ because of dangling edges.  
\begin{defi}[Events happening together] Let $e_1 = (L_1 \overset{l_1}{\hookleftarrow} I_1 \overset{r_1}{\hookrightarrow} R_1, m_1 : L_1 \rightarrow G)$, and $e_2 = (L_2 \overset{l_2}{\hookleftarrow} I_2 \overset{r_2}{\hookrightarrow} R_2, m_2 : L_2 \rightarrow G)$ be events happening on the same graph. Then, $e_1$ and $e_2$ can happen together if $(L_1 \times_{G} L_2) \times_{L_1} I_1 \cong (L_1 \times_{G} L_2) \times_{L_2} I_2$, i.e. there exists a unique $B$ up to isomorphism such that the following diagram commutes and the top 2 squares are each pullback squares.
    \begin{center}
        \begin{tikzcd}
I_1 \arrow[d, "l_1"]  & B \arrow[l] \arrow[r] \arrow[d]                       & I_2 \arrow[d]          \\
L_1 \arrow[rd, "m_1"] & L_1 \times_{G} L_2 \arrow[l, "p_1"'] \arrow[r, "p_2"] & L_2 \arrow[ld, "m_2"'] \\
                      & G                                                     &                       
\end{tikzcd}
    \end{center}
\end{defi}
\noindent The existence of $B$ means that $I_1\cap L_2 = I_2 \cap L_1 = B$ in $G$. This means that if some vertex/edge is removed in one rule, then it cannot be preserved by the other. A diagram that illustrates this definition is given in figure \ref{schematic}.
\begin{figure}
    \centering
    \includegraphics[scale = 0.25]{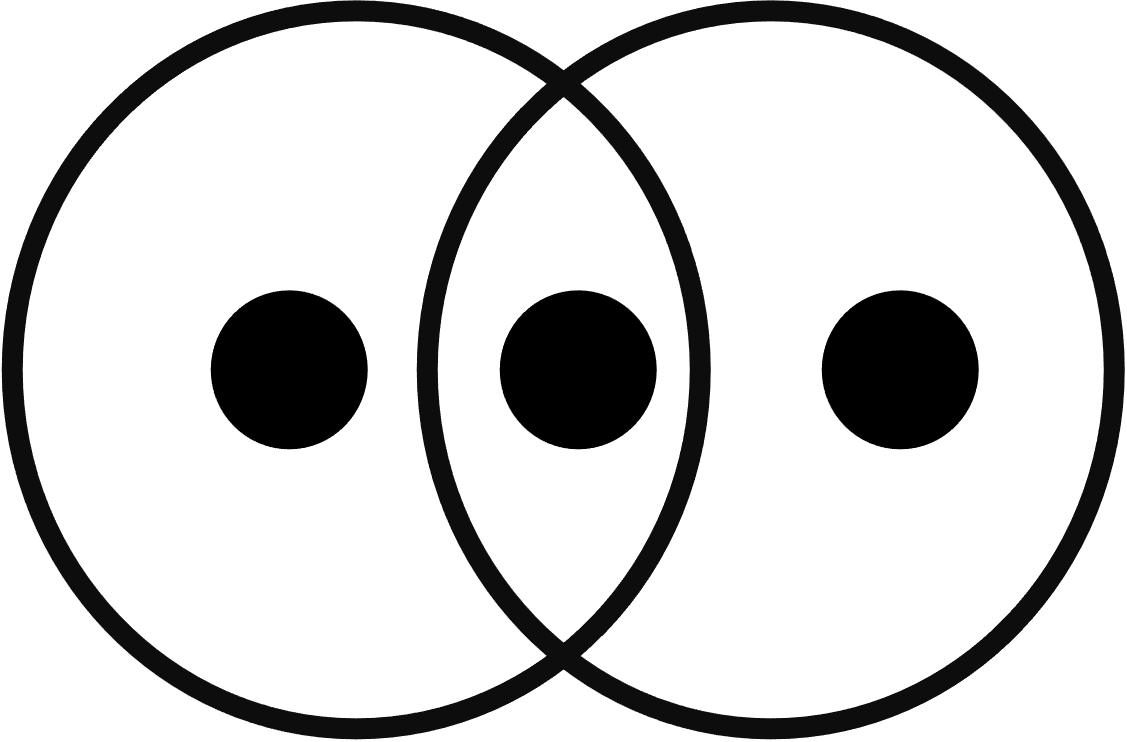}
    \caption{2 events that can happen together}
    \label{schematic}
\end{figure}
The application areas of the rules might intersect, and the shaded areas represent the parts of the graph that are preserved by the rule (the interface graphs). We can create the combined event $L = L_1 \sqcup_{L_1 \times_{G} L_2} L_2 \hookleftarrow I = I_1 \sqcup_{B} I_2 \hookrightarrow R  = R_1 \sqcup_B R_2$ with the unique match $m$ that makes the following diagram commute:
\begin{center}
   \begin{tikzcd}
L_1 \times_{G} L_2 \arrow[r, "p_2"] \arrow[d, "p_1"'] & L_2 \arrow[d] \arrow[rdd, "m_2", bend left] &   \\
L_1 \arrow[r] \arrow[rrd, "m_1"', bend right]         & L \arrow[rd, dashed]                        &   \\
                                                      &                                             & G
\end{tikzcd}
\end{center}
 Note that this notion of concurrency in graph rewriting is different from the one usually in the literature surrounding algebraic graph transformation theory \cite{Adhesive}:
\begin{defi}[Parallel-independent events]
   Let $e_1 = (L_1 \overset{l_1}{\hookleftarrow} I_1 \overset{r_1}{\hookrightarrow} R_1, m_1 : L_1 \rightarrow G)$, and $e_2 = (L_2 \overset{l_2}{\hookleftarrow} I_2 \overset{r_2}{\hookrightarrow} R_2, m_2 : L_2 \rightarrow G)$ be events happening on the same graph. Then, they are \textit{parallel-independent} if there exist morphisms $r : L_{1} \rightarrow G_{2}$, $s : L_{2} \rightarrow G_{1}$ such that the following diagram commutes:
   \begin{center}
       \begin{tikzcd}
R_{1} \arrow[d] & I_{1} \arrow[l] \arrow[d] \arrow[r] & L_{1} \arrow[rd] \arrow[rrrd, "r"] &   & L_{2} \arrow[ld] \arrow[llld, "s"'] & I_{2} \arrow[l] \arrow[d] \arrow[r] & R_{2} \arrow[d] \\
H_{1}           & G_{1} \arrow[l] \arrow[rr]          &                                    & G &                                     & G_{2} \arrow[ll] \arrow[r]          & H_{2}          
\end{tikzcd}
   \end{center}
\end{defi}
\noindent This says that there is no interference between the events i.e. there are no vertices that are deleted by the one but used by the other. Thus, 2 parallel-independent events can happen sequentially. However, if $e_{1}$ and $e_{2}$ \textit{can happen together} (by our definition) and there is a non-empty subgraph that is deleted by both, then they cannot occur sequentially! In the terminology used by Winskel \cite{EventStructures}, these events are not consistent because they cannot occur in the same history. However, at the same time, there is a degree of compatibility between the events that it is possible for them to occur simultaneously. This indicates that a more detailed analysis of the algebraic structure of graph rewriting is needed. 
\begin{note} 
    If event $e_1$ can happen with $e_2$, and $e_2$ can happen with $e_3$, then it doesn't imply that $e_1$ can happen with $e_3$. 
\end{note}
\begin{defi}[$n$ events happening together]
    A collection of $n$ events $e_i = (L_i \overset{l_i}{\hookleftarrow} I_i \overset{r_i}{\hookrightarrow} R_i, m_i : L \rightarrow G)$ can all happen together if they can happen together pairwise i.e. $e_i$ is concurrent with $e_j$ for any $i \neq j$ distinct.
\end{defi}
\noindent Since finite limits exist (there is a terminal object and pullbacks exist), let $L$ depict the limit as shown in the commutative diagram below:
\begin{center}
\begin{tikzcd}
                &                    & L \arrow[ld] \arrow[d] \arrow[rd] \arrow[lld] \arrow[rrd] &                    &                 \\
... \arrow[rrd] & L_{i-1} \arrow[rd] & L_i \arrow[d]                                             & L_{i+1} \arrow[ld] & ... \arrow[lld] \\
                &                    & G                                                         &                    &                
\end{tikzcd}
\end{center}
Consider the pullback square 
\begin{center}
    \begin{tikzcd}
A_i \arrow[r] \arrow[d] & L \arrow[d] \\
I_i \arrow[r, "l_i"']   & L_i        
\end{tikzcd}
\end{center}

\noindent where $A_i$ can be thought of as $I
_i \cap L$, which is the subgraph that is preserved in rule $i$, and is used by the other rules.  If these events can happen together, we would want that $A_i$ would be the same as $A_j$ for distinct $i,j$ because if not, there will be vertices/edges deleted by one rule that are preserved by the other. Consider the pullback $L_i \times_{L} L_j$. We have a unique morphism $L \rightarrow L_i \times_{L} L_{j} = L_i \cap 
 L_j$. We get the following commutative diagram
\begin{center}
    \begin{tikzcd}
               &                                                                             & B_{ij} \arrow[rd] \arrow[lldd, bend right] &                                      \\
               & B_{ij}\cap L \arrow[ru] \arrow[rd] \arrow[ld] \arrow[d, dashed, bend right] &                                            & L_i \cap L_j \arrow[lldd, bend left] \\
I_i \arrow[rd] & A_i \arrow[l] \arrow[r] \arrow[u, dashed, bend right]                       & L \arrow[ru] \arrow[ld]                    &                                      \\
               & L_i                                                                         &                                            &                           
\end{tikzcd}
\end{center}
where $B_{ij} = I_i \cap L_j$, which is the same as $I_i \cap L_j$ by assumption. Here, $B_{ij} \cap L$ is the pullback of $B_{ij} \rightarrow L_i \cap L_j$. We get that $B_{ij} \cap L \rightarrow I_i$, which is the composition of $B_{ij} \cap L \rightarrow B_{ij}$ and $B_{ij} \rightarrow I_i$, and $B_{ij} \cap L \rightarrow L$. Then by the universal property of pullbacks, there is a unique $B_{ij} \cap L \rightarrow A_i$. There is a unique $A_i \rightarrow B_{ij}$ since $B_{ij}$ is a pullback, and we have $A_i \rightarrow L$. Thus, there is a unique $A_i \rightarrow B_{ij} \cap L$. We can show via universal properties that the composition of $A_i \rightarrow B_{ij} \cap L$  Thus, $A_i \cong B_{ij} \cap L$. Similarly, we will get that $A_j \cong B_{ij} \cap L$. So $A_i \cong A_j$. We call $e_1 = (L_1 \overset{l_1}{\hookleftarrow} I_1 \overset{r_1}{\hookrightarrow}, m_1 : L_1 \rightarrow G)$, $e_2 = (L_2 \overset{l_2}{\hookleftarrow} I_2 \overset{r_2}{\hookrightarrow}, m_2 : L_2 \rightarrow \Tilde{G})$ to be \textit{successive} if $\tilde{G}$ is the production of $e_{1}$. We can now define what it means for 2 events to be causally related:
\begin{defi}[$\leq_{C}$]
    Let $e_1 \leq e_2$ as above. Then, $e_1 \leq_{\mathcal{C}} e_2$ if the pullback of $L_2 \overset{m_2}{\rightarrow} \Tilde{G}$ and $G' \rightarrow \Tilde{G}$ is not $L_{2}$, where $G'$ is the pushout complement of $I_{1} \overset{l_{1}} {\hookrightarrow} L_{1} \overset{m_{1}}{\rightarrow} G$. 
\end{defi}
\noindent If the pullback is $L_{2}$, then it means that the image of $m_{2}$ is in the pushout complement $G'$, and therefore in $G$, which means that the second event could have occurred before. If the pullback is not $L_{2}$, then it means that the second event uses \textbf{new} vertices/edges produced by the first as input. Note that if 2 successive events are causally disconnected, then there is a morphism $m_2' : L_2 \rightarrow G'$ such that the following is a pullback square
\begin{center}
    \begin{tikzcd}
L_{2} \arrow[d, "m_2'"'] \arrow[r, "\text{id}"] & L_2 \arrow[d, "m_2"] \\
G' \arrow[r, "\iota"']                          & G                   
\end{tikzcd}
\end{center}
\noindent Since $G' \overset{\iota}{\rightarrow} G$ is a monomorphism, we have a matching $\iota \circ m_2' : L_2 \rightarrow G$. However, it is not necessary that these events can happen together. This is because although event 1 does not delete anything in $L_2$ (which is why $L_2$ is in the pushout complement), it may very well be the case that event 2 deletes part of $L_1$. If they can occur together, then the order of events may be switched. This is called the \textit{local church-rosser} property in algebraic graph transformation theory. 
\section[Causality in lambda]{Causality in $\lambda-$calculus}
\subsection{Overview of previous work}
The notion of causality is well-studied in both theoretical computer science and theoretical physics (most prominently in Einstein's relativity). Glynn Winskel's paper titled ``Event Structures" \cite{EventStructures} set the stage for the study of causality in various models of computation. Assuming that one knows the set of events $E$ that occurred during a computation, the causal relations between them is traditionally thought of as a partial order on $E$ satisfying reflexivity, transitivity, and antisymmetry ($a \leq b, b \leq a \Rightarrow a = b)$. We instead don't require reflexivity (on the grounds that an event can't ``cause" itself) and modify the antisymmetric condition:
\begin{defi}[Causal Set]
    A causal set is a pair $(X, \leq_{\mathcal{C}})$ where $X$ is the set of \textit{events} and $\leq_{\mathcal{C}}$ is a relation on $X$ satisfying:
    \begin{itemize}
        \item $(\forall x,y,z) $ $ x \leq_{\mathcal{C}} y, y \leq_{\mathcal{C}} zz$ $ \Rightarrow x \leq_{\mathcal{C}} z$
        \item $(\forall x,y) $ $x \leq_{\mathcal{C}} y \Rightarrow y \not \leq_{\mathcal{C}} x$
    \end{itemize}
\end{defi}
\noindent From the above, it is clear that a causal relation on a set cannot be reflexive. The second relation explicitly says that if $x$ causes $y$, then $y$ cannot cause $x$, which seems to be true in most physical and computer systems (we will discuss the notion of causality in physics and its relation to the one we use in this paper later on). Winskel expanded on the definition of a causal set to distinguish between concurrent and causally disconnected events. Namely, two concurrent events must be causally disconnected, but not necessarily the other way around. For instance, there can exist events $A$ and $B$ which may use some common ``resource" i.e. both cannot occur at the same time, but neither \textit{causes} the other. In the case of hypergraph rewriting, for example, 2 events that can happen together but have a non-empty overlap of interface graphs are causally disconnected, but they cannot occur sequentially because one deletes some vertices/edges that are in the matching of the other event. So they are not concurrent. In \cite{Joyal}, concurrency is defined as the following:
\begin{defi}[Concurrent events]
    Two events $e,e'$ are concurrent if they are not causally related (i.e. $e \not \leq_{\mathcal{C}} e'$ and $e' \not \leq_{\mathcal{C}} e)$, and they are \textit{consistent} i.e. they can happen in the same computation.
\end{defi}
\noindent It is easy to see from this definition that if $a \overset{\beta}{\rightarrow} b \overset{\beta}{\rightarrow} c$ are successive transitions, then they are causally disconnected iff they are concurrent (because they are already consistent). Winskel 
\cite{EventStructures} used the notion of an ``event structure'' as a framework to dealing with events, their causal relationships, and concurrency. We first look at the notion of a prime event structure \cite{Winskel2}
\begin{defi}[Prime event structure]
    A prime event structure is a structure $E = (E, \#, \leq)$ consisting of a set $E$ of events which are partially ordered by $\leq$, the \textit{causal dependency relation}, and a binary, symmetric, irreflexive relation $\# \subseteq E \times E$, the \textit{conflict} relation, satisfying:
    \begin{itemize}
        \item $\lceil e \rceil = \{e' : e' \leq e\}$ is finite
        \item $e \# e', e' \leq e'' \implies e \# e''$ 
    \end{itemize}
\end{defi}
\noindent The first condition is called the \textbf{axiom of finite causes}, which states that any event occurs after a finite number of steps ($\beta-$reductions), starting from some initial state. The conflict relation specifies which events cannot occur in the same computational thread. Note that if $e_1 \leq e_2$ and $e_1 \# e_2$, then by symmetry of $\#$, $e_2 \# e_1$, and by the second axiom, $e_1 \# e_1$, which is a contradiction. This makes sense as $e_1 \leq e_2$ implies that $e_1$ and $e_2$ can occur in the same computational process, which is why we can make a statement about their causal dependency. 
\color{black}
\begin{defi}[Event Structure]
    An \textit{Event Structure} is a triple $(E, \text{Con}, \vdash)$ where
    \begin{enumerate}
        \item $E$ is a set of events
        \item Con is a non-empty subset of Fin$E$ (the set of finite subsets of $E$), called the consistency predicate, which satisfies 
        \begin{equation*}
            X \in \text{Con}, Y \subseteq X \Rightarrow Y \in \text{Con}
        \end{equation*}
        \item $\vdash \subseteq \text{Con} \times E$ is the \textit{enabling relation} which satisfies $X \vdash e$, $X \subseteq Y \in \text{Con} \Rightarrow Y \vdash e$
    \end{enumerate}
\end{defi}
\noindent The second condition states that if a finite set $X$ of events can all occur in the same computational history, then a subset $Y$ of those events can also occur in the same history, which is an obvious statement. The requirement that all elements of Con are finite subsets is motivated from the axiom of finite causes. The enabling relation is thought of as a replacement for the causal dependency relation. 
\color{black}
\newline \\
\noindent The notion of causality is also well-studied in theoretical physics. A brief summary is provided in \cite{Prakash}. In Einstein's theory of relativity, points in spacetime are called ``events" and 2 events, $a = (t_a, x_a)$ and $b = (t_b, x_b)$, are causally related if an observer located at $x_a$ can send a signal at time $t_a$ towards the observer at $x_b$ so that it reaches $x_b$ before $t_b$. The maximum speed of any signal is the speed of light. Thus, if $t_b > t_a$, then $a \leq_{\mathcal{C}} b$ if $\frac{|x_b - x_a|}{t_b - t_a} \leq c$. The idea is that a signal may potentially contain information that can ``influence" event $b$, which means there might be a causal relationship between them. It can easily be checked that spacetime (any, flat or curved) satisfies the axioms of a causal set. If 2 events have a large separation in space but a small one in time such one cannot influence the other by means of a signal, then they are called ``spacelike separated" (because they are more separated in space than in time). 
\color{black}
\newline \\
\noindent We would like to discuss the notion of events and their causal relationships in the context of $\lambda-$calculus. The syntax is defined as usual by $M = x | \lambda x. M | (M N)$, the last two being \textit{abstraction} and \textit{application}. The production rules for $\lambda-$calculus that define the $\beta-$reduction relation on $\lambda-$expressions are the following:
\begin{center}

\begin{tabular}{cc}
    \begin{bprooftree}
    \AxiomC{}
    \UnaryInfC{$(\lambda x.M)N \rightarrow_{\beta} M[N/x]$}
  \end{bprooftree} &
  \begin{bprooftree}
    \AxiomC{$M \rightarrow_{\beta} M'$}
    \UnaryInfC{$\lambda x.M \rightarrow_{\beta} \lambda x.M'$}
  \end{bprooftree} \\[2em]
  \begin{bprooftree}
      \AxiomC{$M \rightarrow M'$}
      \UnaryInfC{$(M N) \rightarrow_{\beta} (M' N)$}
  \end{bprooftree} &
    \begin{bprooftree}
      \AxiomC{$N \rightarrow N'$}
      \UnaryInfC{$(M N) \rightarrow_{\beta} (M N')$}
  \end{bprooftree}
\end{tabular}

\end{center}
There are many excellent references to learn more about the $\lambda-$calculus, and the reader is encouraged to look at them. 
\subsection[good]{\textit{Good} $\lambda-$expressions
}
\noindent Let's try to understand causality in the context of $\lambda-$calculus. Assume we have a reduction $(\lambda x.(x) a)(\lambda y.y) \rightarrow_{\beta} (\lambda y.y)a \rightarrow_\beta a$. Considering these $\beta-$reductions as ``events", the statement that the first event \textit{causes} the second makes sense. This is because in the first expression, the application $(x)a$ cannot be reduced as $x$ is not a $\lambda-$abstraction. The first event puts $(\lambda y.y)$ in the place of $x$ and allows the inner reduction to occur. According to the grammar of $\lambda-$calculus, every reduction of an expression $e$ is an application of an abstraction $\lambda x. M$ to some argument $e'$ (where $e'$ is any $\lambda-$expression), where $(\lambda x. M)e' \subseteq e$. In a rough sense, to determine the causal relationships between events, one must keep track of every subexpression of the form $(e e')$ during an evaluation ($e$ is typically called the function/abstraction and $e'$ the argument) Then, one can identify the $\beta-$reductions that are responsible for $e$ eventually being replaced by $\lambda-$abstraction. How can this be accomplished?   
\\
\newline
\noindent We can add metadata to the $\lambda-$expression enumerating all possible applications. For instance, in the expression $(\lambda x.x \otimes_{1} a) \otimes_{2} (\lambda y.y)$, $\otimes_{1}$ refers to the application $(x) a$ and $\otimes_{2}$ refers to the entire $\lambda-$expression. The reduction sequence until termination is $(\lambda x.(x) \otimes_{1} a) \otimes_{2} (\lambda y.y)$ $\rightarrow$ $(\lambda y.y) \otimes_{1} a$ $\rightarrow$ $a$. The first and second $\beta-$reductions can be labelled $e_2$ and $e_1$ respectively, where $e_2 \leq_{\mathcal{C}} e_1$. It seems that the same reasoning may be used to determine whether 2 successive events are causally related. However, determining causal relationships between successive events is not sufficient to determine those between events separated by intermediate steps. An example is given in figure $\ref{non-timelike}$.  
\begin{figure}
    \centering
    \includegraphics[scale = 0.6]{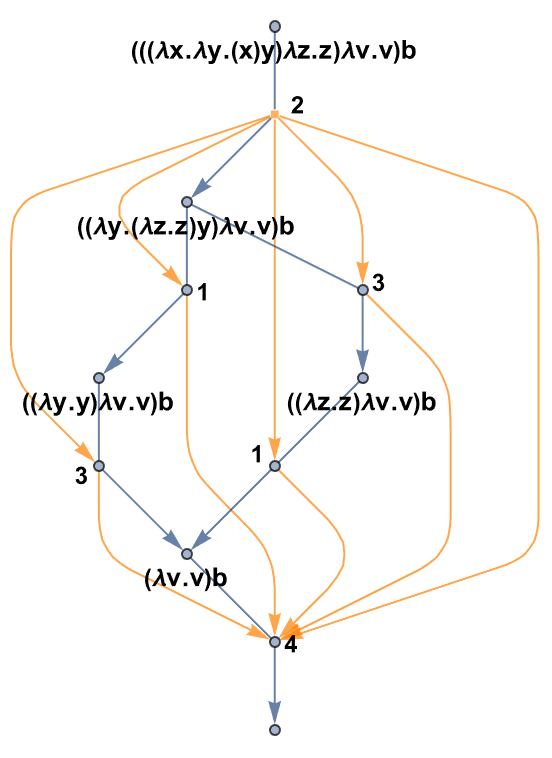}
    \caption{The causal graph for the expression $(((\lambda x.\lambda y. (x) y)\lambda z.z)\lambda v.v)b$}
    \label{non-timelike}
\end{figure}
\\
\newline
\noindent In the figure, the $\lambda-$expression is labelled as $(((\lambda x.\lambda y. (x) \otimes_{1} y) \otimes_{2}\lambda z.z) \otimes_{3} \lambda v.v) \otimes_{4} b$. In the left-most path of the computation, event 1 is causally related to event 4 despite not being successive. Moroever, taking transitive closures of causal relations between successive events does not yield that $1 \leq_{\mathcal{C}} 4$, which is true  because the reduction of the sub-expression $(\lambda z.z)y \rightarrow_{\beta} y$ is necessary to arrive at a state where event $4$ can be applied. We can reason about this by looking at the second path where events $1$ and $4$ are successive events and they are causally related. This assumes that causal relationships are preserved in different paths, whereas \textit{a priori} there is no reason why the definition of causality in a computational system should not depend on the path taken. Asking whether 2 events are causally related only makes sense if they lie in the same path. Thus, studying other paths to determine causal relationships in one implies the existence of an ``event labeling" function, as we have above in figure $\ref{non-timelike}$. However, such a function might not exist. For instance, if we have a lambda expression $(\lambda x.(x \otimes_{1} x)) (\lambda y.y \otimes_{2} \lambda y.y)$, the events in p1 cannot be mapped bijectively the the events in the p2:
\begin{equation} \label{badpath1}
(\lambda x.(x \otimes_{1} x)) ((\lambda y.y) \otimes_{2} \lambda z.z) \rightarrow_{\beta} (\lambda x.(x \otimes_{1} x)) (\lambda z.z) \rightarrow_{\beta} (\lambda z.z) \otimes_{1} \lambda z.z \tag{p1}
\end{equation}
\begin{gather*} \label{badpath2}
 (\lambda x.(x \otimes_{1} x)) ((\lambda y.y) \otimes_{2} \lambda z.z) \rightarrow_{\beta} ((\lambda y.y) \otimes_{2} \lambda z.z) \otimes_{1} ((\lambda y.y) \otimes_{2'} \lambda z.z) \rightarrow_{\beta} \\
 (\lambda z.z) \otimes_{1} ((\lambda y.y) \otimes_{2'} \lambda z.z) \rightarrow_{\beta} (\lambda z. z) \otimes_{1} (\lambda z.z) \tag{p2}
\end{gather*}
\noindent  Since $x$ occurs twice, event $2$ is duplicated in the p2, where the label $2'$ denotes the duplicated event. In p1, the argument $(\lambda y.y) \otimes \lambda z.z$ is evaluated before being substituted, preventing duplication of events.
It is clear that event $1$ can be identified in both paths; however, event $2$ in p1 can be identified with either event $2$ or event $2'$ in p2. Thus, there is no canonical way to label events.  Reducing the arguments before substituting them is termed \textit{applicative} order reduction, differing from the opposite strategy of substituting them before, called \textit{normal} order reduction. If the number of free occurrences of a variable $x$ in $M$ exceeds 1, then the normal order reduction of an expression $(\lambda x. M)N$ will take more steps than its applicative order reduction. If, however, for every subexpression of the form $(\lambda x. M) N$, $x$ occurs at most once in $M$, we will show that every possible path will consist of the same events but evaluated in a different order, permitting the construction of an event labeling function. Let $\Lambda$ be a set of labels. Then, we can define the set of \textit{good}-$\lambda$ expressions, denoted $G(\Lambda)$, and $L : G(\Lambda) \rightarrow \text{Fin}(\Lambda)$ recursively:
\begin{defi}[$G(\Lambda)$] If $x$ is a variable, then $x \in G(\Lambda)$ and $L(x) = \emptyset$. If $M \in G(\Lambda)$ and the number of free occurrences of $x$ in $M$ is at most 1, then $\lambda x. M \in G(\Lambda)$ and $L(\lambda x. M) = L(M)$. If $M, N \in G(\Lambda)$ where $L(M) \cap L(N) = \emptyset$ and $l \in \Lambda \setminus (L(M) \cup L(N))$, then $M \otimes_{l} N \in G(\Lambda)$ and $L(M \otimes_{l} N) = L(M) \cup L(N) \cup \{l\}$. 
\end{defi}
\noindent We will take $\Lambda = \mathbb{N}$ for ease of notation. The production rules are identical to those of arbitrary $\lambda-$expressions.
\begin{center}
    \begin{tabular}{cc}
    \begin{bprooftree}
    \AxiomC{}
    \UnaryInfC{$(\lambda x.M) \otimes_{n} N \rightarrow_{\beta} M[N/x]$}
  \end{bprooftree} &
  \begin{bprooftree}
    \AxiomC{$M \rightarrow_{\beta} M'$}
    \UnaryInfC{$\lambda x.M \rightarrow_{\beta} \lambda x.M'$}
  \end{bprooftree} \\[2em]
  \begin{bprooftree}
      \AxiomC{$M \rightarrow M'$}
      \UnaryInfC{$(M \otimes_{n} N) \rightarrow_{\beta} (M' \otimes_{n} N)$}
  \end{bprooftree} &
    \begin{bprooftree}
      \AxiomC{$N \rightarrow N'$}
      \UnaryInfC{$(M \otimes_{n} N) \rightarrow_{\beta} (M \otimes_{n} N')$}
  \end{bprooftree}
\end{tabular}
\end{center}
\noindent It is evident that good lambda expressions are closed under $\beta-$reduction i.e. if $M \in G(\Lambda)$ and $M \rightarrow_{\beta} N$, then $N \in G(\Lambda)$ as well. We define the notion of \textit{length} for a good $\lambda-$expression as the number of events in it, denoted by $l : G(\Lambda) \rightarrow \mathbb{Z}_{\geq 0}$,$ e \mapsto |L(e)|$. A lambda expression is said to be in $\beta-$normal form if it cannot be reduced further.
\begin{defi}[Terminating relation]
    A relation $\rightarrow$ on a set is terminating if there is no infinite descending chain $a_0 \rightarrow a_1 \rightarrow a_2 \rightarrow ..$
\end{defi}
\noindent By the production rules, $M \rightarrow_{\beta} N \Rightarrow l(N) = l(M) + 1 > l(M)$, and because $(\mathbb{Z}_{\geq 0}, >)$ is terminating, $\rightarrow_{\beta}$ is terminating. 
(more formally, $l$ is a \textit{monotone} map \cite{TermsAllThat}). 
\begin{defi}[Diamond Property]
    A relation $\rightarrow$ on a set $A$ satisfies the 
$\textit{diamond}$ property if for any $a,b,c \in A$ such that $a \rightarrow b$, $a \rightarrow c \Rightarrow (\exists d \in A) b \rightarrow d, c\rightarrow d$.  
\end{defi}
\begin{prop}
    $\rightarrow_{\beta}$ on $G(\Lambda)$ satisfies the diamond property.
\end{prop}
\begin{proof}
      This can be proved by structural induction on the good $\lambda-$expressions. Assume $(\lambda x. M) \otimes_{n} N \rightarrow_{\beta} (\lambda x. M') \otimes_{n} N$. We have the diamond $(\lambda x.M') \otimes_{n} N \rightarrow_{\beta} M'[N/x]$ and $M[N/x] \rightarrow_{\beta} M'[N/x]$. Assume $(\lambda x. M) \otimes_{n} N \rightarrow_{\beta} (\lambda x. M) \otimes_{n} N'$. Since the number of free occurrences of $x$ in $M$ is at most 1, $M[N/x] \rightarrow_{\beta} M[N'/x]$ (because there are no duplicates), closing the diamond with the other transition $(\lambda x. M) \otimes_{n} N' \rightarrow M[N'/x]$. Another case to check is $(\lambda x. M) \otimes_{n} N \rightarrow_{\beta} (\lambda x. M') \otimes_{n} N$ and $(\lambda x. M) \otimes_{n} N \rightarrow_{\beta} (\lambda x. M) \otimes_{n} N'$. This obviously closes to form a diamond where both expressions can be reduced to $(\lambda x. M') \otimes_{n} N'$. All the other cases follow by structural induction. 
 \end{proof}
\begin{defi}[Event]
If $A \rightarrow_{\beta} B$, then by the production rules, the reduction is of the form $(\lambda x.M) \otimes_{n} N \rightarrow_{\beta} M[N/x]$ where $(\lambda x.M) \otimes_{n} N $ is a subexpression of $A$. The event of this transition is denoted $(A,B,n) \in G(\Lambda) \times G(\Lambda) \times \Lambda$. The event may be denoted as $A \rightarrow_{\beta, n} B$. Since we can recover the label $n$ from $A,B$, the event might also just be labelled $(A,B)$. 
\end{defi}
\noindent Using the above terminology, we can define whether 2 successive events are causally related or not. If $e_1 = A \rightarrow_{\beta, n} B$ and $e_2 = B \rightarrow_{\beta, m} C$, then $e_1 \leq_{\mathcal{C},2} e_2$ if in $A$, the application subexpression with the label $m$ (we know $m \in L(A)$) is of the form $e \otimes_{m} e'$ where $e$ is not a $\lambda-$abstraction. In $B$, $e$ must be a $\lambda-$abstraction. The fact that the diamond property is obeyed in the reduction of good $\lambda-$expressions gives us the following result:
\begin{prop}
    Assume that $e_1$ and $e_2$ are causally disconnected where $e_1 = A \rightarrow_{\beta,n} B$ and $e_2 = B\rightarrow_{\beta, m} C$. Then, there is some $D$ such that $A \rightarrow_{\beta,m} D$ and $D \rightarrow_{\beta, n} C$. 
\end{prop}
\begin{proof}
    By definition, the application with label $m$ in $A$ is of the form $(\lambda x. M) \otimes_{m} N$, which means event $m$ can occur without $n$. Thus, we have a reduction $A \rightarrow_{\beta,m} D$ where the subexpression $(\lambda x. M) \otimes_{m} N$ is evaluated. By the diamond property, there must be a transition $D \rightarrow_{\beta} C$. Since $n \in L(D)$ but $n \notin L(C)$, the transition from $D$ to $C$ must be the event $n$. 
\end{proof}
\noindent We now define the notion of a multiway system, which consists of all possible paths of evaluation starting from a single $\lambda-$expression. Multiway systems are ubiquitous in the Wolfram Physics Project \cite{WPP}. 
\begin{defi}[Multiway System]
    Let $A$ be a good $\lambda-$expression. Then, the multiway system of $A$, denoted $\mathcal{M}(A)$, is a graph with vertices $B$ such that $A \rightarrow_{\beta}^* B$ (the reflexive transitive closure), and there is an edge between $B$ and $C$ if $B \rightarrow_{\beta} C$.
\end{defi}
\noindent Notice that given an arbitrary good $\lambda$-expression $A$, there are only a finite number of distinct $A_1,...,A_n$ such that $A \rightarrow_{\beta} A_i$ (i.e. $\rightarrow_{\beta}$ is finitely branching), where $n \leq l(A)$. Since $\rightarrow_{\beta}$ is finitely branching and terminating, $\mathcal{M}(A)$ is a finite graph. It is well-known that if a relation satisfies the diamond property, then so does its transitive closure (this fact is used to prove the church-rosser property). Thus, $\rightarrow_{\beta}^+$ (the transitive closure) satisfies the diamond property. Thus, for any $A \in G(\Lambda)$, there is a unique $B$ in $\beta-$normal form such that $A \rightarrow_{\beta}^+ B$ i.e. all paths in $\mathcal{M}(A)$ converge towards a unique $\lambda-$expression $B$. Furthermore, any path must have length $l(A) - l(B)$ because each $\beta-$reduction reduces the length of a $\lambda-$expression by 1. The set of events occurring in each path is $L(A) \setminus L(B)$. Thus, the same set of events occur in every path, the only difference being the order in which they happen. The following lemma relates our local notion of causality for successive events, $\leq_{\mathcal{C},2}$, to a global statement about the multiway system.
\begin{nlemma}
    Let $e = B \rightarrow_{\beta,n} C$, $e' = C \rightarrow_{\beta, m} D$ be 2 successive events in $\mathcal{M}(A)$. Then, $e \leq_{\mathcal{C},2} e'$ iff for all paths in $\mathcal{M}(A)$, event with label $n$ occurs before the event with label $m$.
\end{nlemma}
\begin{proof}
    $(\Leftarrow)$ Assume that they are causally disconnected. Then by the proposition on the previous page, we know that there is a diamond $B \rightarrow_{\beta, m} E \rightarrow_{\beta, n} D$. This lies in another path of the multiway system, and $m$ occurs before $n$. $(\Rightarrow)$ Assume that they are causally related. This means that in $B$, the application with label $m$ is of the form $\text{expr}_{1} \otimes_{m} \text{expr}_{2}$ where $\text{expr}_{1}$ is not a $\lambda-$abstraction. By the grammar, it must be a symbol or an application. If it's a symbol $a$, then the only way for event $m$ to occur in $C$ is if $a$ is a bound variable in a $\lambda-$abstraction of the form $\lambda a. \text{expr}$, where $a \otimes_{m} \text{expr}_{2}$ is a sub-expression of $\text{expr}$. Thus, the sub-expression with application $n$ in $B$ must be $(\lambda a. \text{expr}) \otimes_{n} \lambda y. \text{expr}'$. Now, the only way for $m$ to occur in any path is for $n$ to occur because that's the only way $a$ can be replaced with a $\lambda-$expression. The other case is when $\text{expr}_{1} = \text{expr} \otimes_{k} \text{expr}'$. If $n$ occurs such that $\text{expr}_{1}$ is replaced by a $\lambda-$abstraction, it must be the case that $k = n$ and $\text{expr}$ is a $\lambda-$abstraction. And in any path, $\text{expr}_{1}$ is a $\lambda-$abstraction only when $n$ has been applied. 
\end{proof}
\noindent Now, we are in a position to define a causal structure on $\lambda-$expressions. First, we define the notion of paths in an abstract rewriting system $(\mathcal{A},\rightarrow_{\beta})$.
\begin{defi}[Paths]
    Let $A,B \in \mathcal{A}$. Then the set of paths from $A$ to $B$, denoted $\mathcal{P}(A,B)$ is the set of all sequences of the form $A \rightarrow_{\beta} S_1 \rightarrow_{\beta} ... \rightarrow_{\beta} S_{n-1} \rightarrow_{\beta} S_{n} = B$. The 
    \textit{event set} of this path is $\{(A, S_1), (S_1, S_2),...,(S_{n-1}, S_n)\}$.
    let $\mathcal{P}(\mathcal{A}) = \bigcup_{A,B \in \mathcal{A}} \mathcal{P}(A,B)$. The event set of some $p \in \mathcal{P}$ is denoted $E(p)$, and $E$ can be thought of as a map $E : \mathcal{P}(\mathcal{A}) \rightarrow \textbf{FinSet}$.
\end{defi}
\noindent $\mathcal{P}(\mathcal{A})$ has a partial order structure, where $p_1 \subseteq p_2$ if $E(p_1) \subseteq E(p_2)$. Defining the category of finite causal sets \textbf{CFin} will simplify the description of the causal relation. See \cite{Dribus} for a more in-depth discussion about the category of causal sets in the context of causal set theory in physics.
\begin{defi}[\textbf{CFin}] The objects of this set are causal sets $(X,\leq)$ where $X$ is a finite set and $\leq$ is a causal relation. A morphism $f: (X, \leq_{x}) \rightarrow (Y, \leq_{y})$ is a set map $f: X \rightarrow Y$ such that if $x_1 \leq_{x} x_2$, then $f(x_1) \leq_{y} f(x_2)$.
\end{defi}
\noindent We have the obvious forgetful functor $F : \textbf{CFin} \rightarrow \textbf{FinSet}$, taking $(X,\leq) \mapsto X$. Note that \textbf{CFin} can be endowed with a partial order, where $(X,\leq_{x}) \subseteq (Y, \leq_{y})$ if $X \subseteq Y$ and the inclusion map $\iota : X \rightarrow Y$ is a morphism of causal sets. We can now define what a $\textit{causal structure}$ on an abstract rewriting system is.
\begin{defi}[Causal structure]
    The causal structure on an ARS $(\mathcal{A}, \rightarrow_{\beta})$ is a $\mathcal{C} : \mathcal{P}(\mathcal{A}) \rightarrow \textbf{CFin}$ such that $F\circ \mathcal{C} = E$ and for any $p_1,p_2 \in \mathcal{P}(\mathcal{A})$ such that $p_1 \subseteq p_2 \Rightarrow \mathcal{C}(p_1) \subseteq \mathcal{C}(p_2)$.
\end{defi}
\noindent In other words, a causal structure on an abstract rewriting system is a causal relation on each path such that if path $p$ is contained in path $q$, then the causal relations on $q$ consists of all the causal relations on $p$. We now introduce the notion of \textit{proper time}, which will be imperative to define the causal structure on $(G(\Lambda),\rightarrow_{\beta})$. 
\begin{defi}[Proper time] Let $A \in G(\mathbb{N})$. Then, let $n,m \in L(A)$ such that for every path $p \in \mathcal{M}(A)$, the event with label $n$ occurs before the event with label $m$. Then, the \textit{proper time} between $n$ and $m$ is $\text{min}_{p \in \mathcal{M}(A)} d_p(n, m)$ where $d_{p}(n, m)$ is 1 + number of events between the event labelled $n$ and $m$ on path $p$.
\end{defi}
\noindent In (general) relativity, proper time between 2 points on a timelike curve $p$ is the time difference between those points as measured by an observer travelling along $p$. The key point is that it is a Lorentz scalar, being independent of the coordinate system used. Here, different paths through $\mathcal{M}(A)$ can be thought of as being different coordinate systems, each providing a different ordering of events. The above definition is obviously independent of the path/coordinate system because it is the minimum separation between 2 events across all paths/coordinate systems. There is a deeper motivation for this definition. In relativity, the causal structure of spacetime is independent of the coordinate system used. We should expect something similar over here. If $A \overset{e_{1}}{\rightarrow_{\beta}} A_1 \overset{e_{2}}{\rightarrow_{\beta}}...\overset{e_{n-1}}{\rightarrow_{\beta}} A_{n-1} \overset{e_{n}}{\rightarrow_{\beta}} A_{n}$ is a path and $e_{n-1}$, $e_{n}$ are causally disconnected, we can construct a new path $A \overset{e_{1}}{\rightarrow_{\beta}} A_1 \overset{e_{2}}{\rightarrow_{\beta}}...\overset{e'_{n-1}}{\rightarrow_{\beta}} A'_{n-1} \overset{e'_{n}}{\rightarrow_{\beta}} A_{n}$ where we have interchanged the order of $e_{n-1}, e_{n}$. If some event $e_{i} \leq_{\mathcal{C}} e_{n}$, for some $1 \leq i \leq n-1$, we expect that $e_{i} \leq_{\mathcal{C}} e'_{n-1}$ in the new path. This is because, in both paths, we can view $e_{n-1}$ and $e_{n}$ as happening together. Unlike in hypergraph rewriting, two successive causally disconnected events in $\lambda-$calculus are non-interfering. Thus, their order of occurrence is irrelevant to causality. This motivates the following definition:
\begin{defi}[Homotopy]
    Let $p_1 = A_0 \rightarrow_{\beta} A_{1} \rightarrow_{\beta} ... \rightarrow_{\beta} A_{n-1}\rightarrow_{\beta} A_{n}$ and $p_{2} = A_{0} \rightarrow_{\beta} A'_{1} \rightarrow_{\beta} ... \rightarrow_{\beta} A'_{n-1} \rightarrow_{\beta} A_{n}$ be paths in $\mathcal{P}(G(\mathbb{N}))$. Then, $p_1 \approx p_2$ if $(\exists i), 1 \leq i \leq n-1$ such that $A_{i} \neq A'_{i}$ and  $(\forall j \neq i)$ $A_{j} = A'_{j}$. Let $\sim$ be the smallest equivalence relation containing $\approx$. Then, $p$ and $q$ are \textit{homotopic} if $p \sim q$. 
\end{defi}
\noindent Note that if events $(A_{i-1}, A_{i})$ and  $(A_{i}, A_{i+1})$ have labels $n,m$ respectively, then events $(A_{i-1}, A'_{i})$ and $(A'_{i}, A_{i+1})$ have labels $m,n$ respectively (the events have been interchanged). We want the causal relation to be invariant under homotopies.
\begin{defi}[Homotopy invariant causal structure]
    A causal structure $\mathcal{C} : \mathcal{P}(G(\mathbb{N})) \rightarrow \textbf{CFin}$ is \textit{homotopy$-$invariant} if whenever $p_1 \approx p_2$, where $p_1 = A_0 \rightarrow_{\beta} ... \rightarrow_{\beta} A_{n-1} \rightarrow_{\beta} A_{n}$ and $p_2 = A_0 \rightarrow_{\beta} A_{1} \rightarrow_{\beta} ... \rightarrow_{\beta} A_{n-1}\rightarrow_{\beta} A_{n}$ and $A_i \neq A'_{i}$, then the map from $E(p_1)$ to $E(p_2)$ where $(A_{i-1}, A_{i}) \mapsto (A'_{i}, A'_{i+1})$, $(A_{i}, A_{i+1}) \mapsto (A'_{i-1}, A'_{i})$, and $(A_{j-1}, A_{j}) \mapsto (A'_{j-1}, A'_{j})$ for all $j \neq i, i+1$ is an isomorphism of causal sets (in \textbf{CFin}).
\end{defi}
\noindent 

\begin{nlemma} Let $\mathcal{C}$ be a homotopy-invariant causal structure on $(G(\mathbb{N}), \rightarrow_{\beta})$ agreeing with $\leq_{\mathcal{C},2}$ on paths of length 2. Let $A \in G(\mathbb{N})$ and assume $n,m \in L(A)$ be events that occur in the reduction of $A$ such that $n$ occurs before $m$ in every path. Then, the subpath $p$ in $\mathcal{M}(A)$ for which the distance between $n,m$ is minimum (i.e. the proper time path) must be of the form
$A_{1} \overset{e_1}{\rightarrow_{\beta}} A_{2} \rightarrow_{\beta}...\rightarrow_{\beta} A_{l} \overset{e_{l}}{\rightarrow_{\beta}} A_{l+1}$ ($e_{1}$ has label $n$ and $e_{l}$ has label $m$) where $e_1 \leq_{\mathcal{C}} e_{i}$ for all $i \in [2, l]$ and $e_{j} \leq_{\mathcal{C}} e_{l}$ for all $j \in [1, l-1]$.  
\end{nlemma}
\begin{proof}
For contradiction, let $i$ be the smallest index such that $e_1 \not \leq_{\mathcal{C}} e_{i}$. If $i  = 2$, then we can interchange the order of events $e_1$ and $e_2$ (because $\leq_{\mathcal{C}}$ agrees with $\leq_{\mathcal{C},2}$ on sucessive events), getting a shorter path between the event with label $n$ and label $m$. If $i \geq 3$, then for all $k \in [2, i-1]$, $e_{1} \leq_{\mathcal{C}} e_{k}$ and $e_{k} \not \leq_{\mathcal{C}} e_{i}$ (otherwise  by transitivity $e_1 \leq_{\mathcal{C}} e_{i}$ which is not true by assumption). Thus we have a path $A_{1} \rightarrow_{\beta} ... \rightarrow_{\beta} A_{i-1} \overset{e_{i-1}}{\rightarrow_{\beta}} A_{i} \overset{e_{i}}{\rightarrow_{\beta}} A_{i+1} \rightarrow_{\beta} ... \rightarrow_{\beta} A_{l+1}$ where $e_{i-1} \not \leq_{\mathcal{C}} e_{i}$. Interchanging the order of events, we get a new path $A_{1} \rightarrow_{\beta} ... \rightarrow_{\beta} A_{i-1} \overset{e'_{i-1}}{\rightarrow_{\beta}} A'_{i} \overset{e'_{i}}{\rightarrow_{\beta}} A_{i+1} \rightarrow_{\beta} ... \rightarrow_{\beta} A_{l+1}$, where $e'_{i-1}$ has the same label as $e_{i}$ and $e'_{i}$ has the same label as $e_{i-1}$. By homotopy invariance, in the the new path, $e_{1} \leq_{\mathcal{C}} e_{k}$ for all $k \in [2,l-2]$ and $e_{1} \not \leq_{\mathcal{C}} e'_{i-1}$. Now, we repeat this procedure until $e_{1}$ and the event with the same label as $e_{i}$ are successive, in which case we can interchange the order again, getting a shorter path between the event with label $n$ and $m$, contradicting minimality. The proof of $e_{j} \leq_{\mathcal{C}} e_{l}$ for all $j \in [1,l-1]$ is exactly the same.
\end{proof}
\begin{thm} There is a unique homotopy$-$invariant causal structure $\mathcal{C} : \mathcal{P}(G(\mathbb{N})) \rightarrow \textbf{CFin}$ which agrees with $\leq_{\mathcal{C},2}$ on paths of length 2. 
\end{thm}
\begin{proof}
    For ease of notation, let $\leq_{\mathcal{C}}$ denote a causal relation induced by $\mathcal{C}$. Let $A \rightarrow_{\beta, n} B$ and $C \rightarrow_{\beta, m} D$ such that $B \rightarrow_{\beta}^{*} C$. Let $p \in \mathcal{P}$ be the path $A \rightarrow_{\beta, n}  B \rightarrow_{\beta} ... \rightarrow_{\beta} C \rightarrow_{\beta, m} D$. We claim that $(A, B, n) \leq_{\mathcal{C}} (C,D, m)$ iff for every path $q$ in $\mathcal{M}(A)$, the event with label $n$ occurs before the event with label $m$ on $q$. $(\Rightarrow)$ Assume that $(A,B,n) \leq_{C} (C,D,m)$. We can induct on the path length. If these events are successive ($B = C$) i.e. $A \rightarrow_{\beta,n} B \rightarrow_{\beta,m} C$, then we are done by lemma 5.1. So assume that there are intermediate events between $B$ and $C$. Assume there is some intermediate event $e$ with label $l_{e}$ such that $(A,B,n) \leq_{\mathcal{C}} e \leq_{\mathcal{C}} (C,D,m)$. Then, by induction (since the length of the path between $(A,B,n)$ and $e$ is strictly smaller), in every path the event with label $n$ occurs before the event with label $l_{e}$. Similarly, by induction, in every path the event with label $l_{e}$ occurs before the event with label $m$. Thus, event $n$ occurs before event $m$. So assume there is no intermediate event $e'$ such that $(A,B,n) \leq_{\mathcal{C}} e' \leq_{\mathcal{C}} (C,D,m)$. Let $e$ be any intermediate event. Then we must have either $(A,B,n) \not \leq_{\mathcal{C}} e$ or $e \not \leq_{\mathcal{C}} (C,D,m)$. In either case, from the proof of the lemma above, we know that we can interchange the order of events such that we get a shorter path  between event $n$ and $m$ which are still causally related (because of homotopy invariance). So it follows from induction again. $(\Leftarrow)$ Assume event $n$ occurs before event $m$ for all paths. The proper time $\tau$ between $n$ and $m$ is a positive integer. We can induct on $\tau$. If $\tau = 1$, then there is some $A',B',C'$ such that $A' \rightarrow_{\beta,n} B' \rightarrow_{\beta,m} C'$, and $(A,B,n) \leq_{\mathcal{C}} (C,D,m)$ by lemma 5.1. So assume the proper time is greater than 1. Consider the proper time path $S \rightarrow_{\beta, n} S_{1} \rightarrow_{\beta} ... \rightarrow_{\beta} S_{n} \rightarrow_{\beta, m} S_{n+1}$. Let $e$ be any intermediate event $(S_i, S_{i+1})$ with label $l_{i}$ for $i \in [1,n-1]$. Then by lemma 5.2 above we know $(S_{0}, S_{1}) \leq_{\mathcal{C}} (S_i, S_{i+1}) \leq_{\mathcal{C}} (S_{n}, S_{m})$. The proper time between event $n$ and $l_{i}$ is less than $\tau$, so by induction, $n$ always occurs before $l_{i}$. Similarly, by induction, $l_{i}$ always occurs before $m$. So $n$ always occurs before $m$.
\end{proof}
\noindent The unique homotopy$-$invariant causal relation on good $\lambda-$expressions is denoted $\leq_{\mathcal{C}}$, and is defined as follows by the above theorem.
\begin{defi}[$\leq_{\mathcal{C}}$]
    Assume that $e = A \rightarrow_{\beta, n} B$, $e' = C \rightarrow_{\beta, m} D$ be 2 events such that there is a path between $B$ and $C$ i.e. $B \rightarrow_{\beta}^+ C$ (or $B = C$). Then, $e \leq_{\mathcal{C}} e'$ if, for all paths in $\mathcal{M}(A)$, the event labeled $n$ occurs before the event labeled $m$. 
\end{defi}
\noindent By demanding homotopy invariance, we have shown that there exists a unique causal structure on the good lambda expressions. The uniqueness comes from the fact that knowing the causal relation 2 successive events is enough to determine the causal relation between arbitrary events. Figure \ref{goodexample} shows the graph of causal relationships between events in the Multiway system of $(((\lambda x. \lambda y. (x) y) \lambda z.z)(\lambda u. \lambda v. v) a) b$.
\begin{figure}
    \centering
    \includegraphics[scale=0.5]{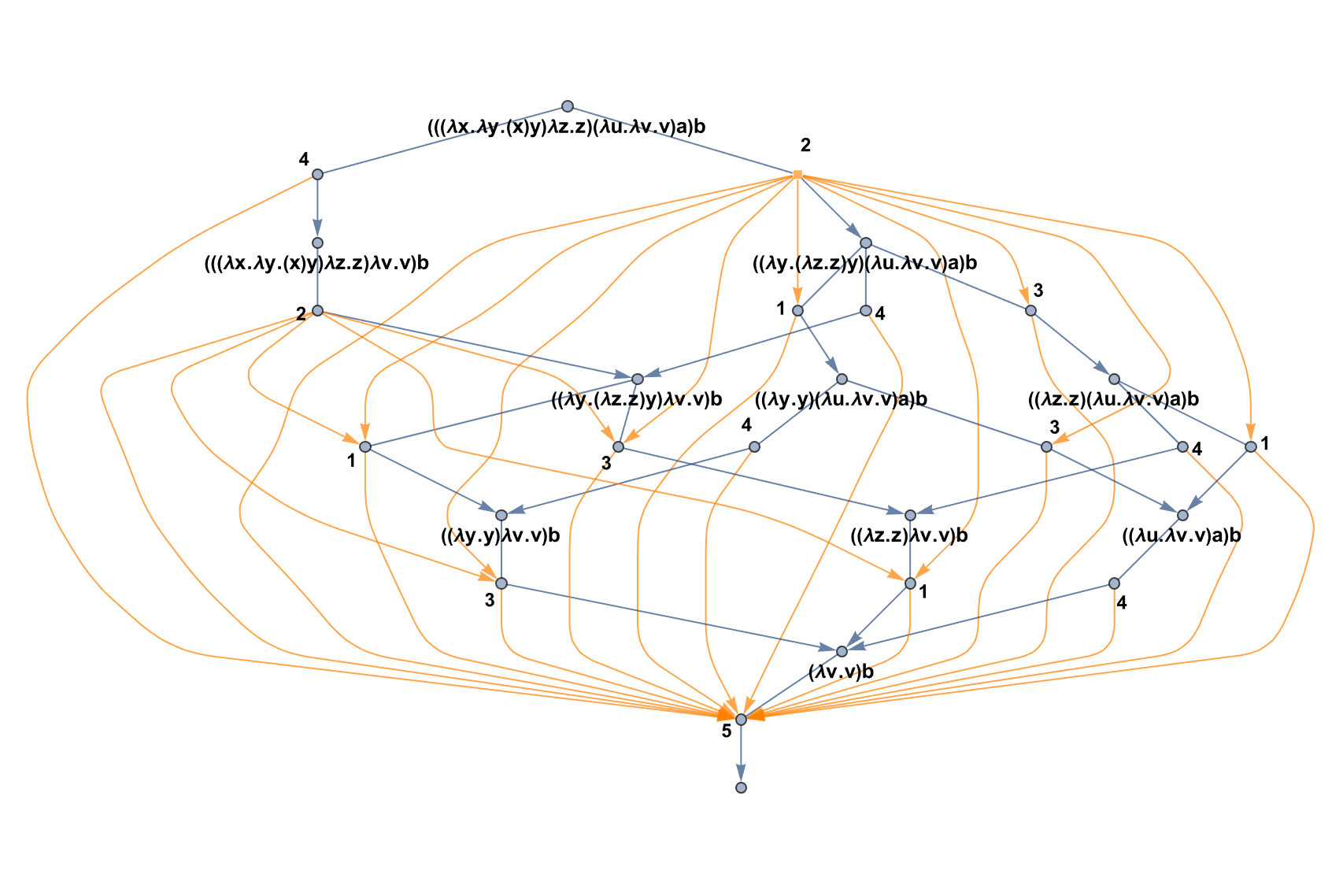}
    \caption{Causal graph for a good lambda expression}
    \label{goodexample}
\end{figure}
We now proceed to discuss causality in the full $\lambda-$calculus using the insights developed in this section.
\subsection[All]{All $\lambda-$expressions} 
Assume that we have a $\lambda-$expression $(\lambda x. M) N$ where the number of free occurrences of $x$, say $n$, in $M$ exceeds 1. Let $(\lambda x. M) N \rightarrow_{\beta} M[N / x] \rightarrow_{\beta} E$ be successive events. Assume that the second event occurs at the sub-expression $(\lambda y. A)(B)$ and $B$ is substituted in the second event.  If $(\lambda y. A)(B)$ is a sub-expression of $M[N/x]$, then we have 4 separate cases: (1) it is a sub-expression of $M$, (2) it is a sub-expression of $N$, (3) $N = B$ and $(\lambda y. A) x$ is a sub-expression of $M$, or (4) $N = (\lambda y. A)$ and $(x \, B)$ is a sub-expression of $M$ . It can be checked from the grammar of $\lambda-$calculus that there are no other cases. In cases (1),(2), and (3), the events are causally disconnected, and in (4), the events are causally related. There is a slightly more elegant way to look at this scenario. Let the free occurrences of $x$ in $M$ be numbered $1,2,..,n$, and let $M[N/x,i]$ denote the substitution of $N$ at the $i$\textit{th} free occurrence of $M$, Performing the substitution one at a time, we can get the new reduction sequence $(\lambda x. M) N \rightarrow_{\beta} (\lambda x. M[N/x,1])N \rightarrow_{\beta} (\lambda x. (M[N/x,1])[N/x,2])N \rightarrow_{\beta} ... \rightarrow_{\beta} M[N/x] \rightarrow_{\beta} E$. Assume we are in case (4), let $j$ be the number assigned to $x$ in $(x \, B)$. Then, we can reorder the order of substitutions as $(\lambda x. M)N \rightarrow_{\beta} (\lambda x. M[N/x,1])N \rightarrow_{\beta} ... \rightarrow_{\beta} (\lambda x. M[N/x,1]...[N/x,j-1])N \rightarrow_{\beta} (\lambda x. M[N/x,1]...[N/x,j-1][N/x,j+1])N \rightarrow_{\beta} ... \rightarrow_{\beta} (\lambda x. M[N/x,1]...[N/x,j-1][N/x,j+1]..[N/x,n])N \rightarrow_{\beta} M[N/x]$, where the last substitution is the $jth$ substitution. Then, in the reduction sequence $(\lambda x. M[N/x,1]...[N/x,j-1][N/x,j+1]..[N/x,n])N \rightarrow_{\beta} M[N/x] \rightarrow_{\beta} E$, since the first expression has only one free occurrence of $x$. In general, it makes sense to say that $((\lambda x. M)N, M[N/x])$ causes  $(M[N/x], E)$ if there exists a $j$ such that $((\lambda x. M[N/x,1]...[N/x,j-1][N/x,j+1]..[N/x,n])N,M[N/x])$ causes $(M[N/x], E)$. Since each substitution is causally disconnected from the other (illustrated by the fact that we can switch the order of substitutions), if we can demand homotopy invariance, then our definition reduces to $((\lambda x. M)N, M[N/x])$ causes  $(M[N/x], E)$ if any one of the events in the path $(\lambda x. M) N \rightarrow_{\beta} (\lambda x. M[N/x,1])N \rightarrow_{\beta} (\lambda x. (M[N/x,1])[N/x,2])N \rightarrow_{\beta} ... \rightarrow_{\beta} M[N/x]$ causes $(M[N/x], E)$. 
\newline
\\
We can label the $\lambda$-expressions similarly to how we did in the preceding section but with a few differences. If we have the expression $(\lambda x. (x \otimes_{1} x)) \otimes_{2} ((\lambda y. y) \otimes_{3} z) $, then the application of event 2 can yield $((\lambda y. y) \otimes_{(3,(2,1))} z)  \otimes_{1} ((\lambda y. y) \otimes_{(3,(2,2))} z)$, where the addition of $(2,1)$ to the label indicates that event 2 occurred, and the expression was substituted at the first occurrence of $x$ (similarly $(2,2)$ indicates that the expression was substituted at the second occurrence of $x$). We are annotating 
$\lambda$-expressions with labels, which are updated with information that encode the reduction event and the substitution location. The labels assigned to the free occurrences of a bound variable must be updated too because in some expression $(\lambda x. M)\otimes N$, new copies of $x$ can be created in $M$ if $M$ is reduced.  Take, for example, $(\lambda x.((\lambda y. (y \otimes_{1} y)) \otimes_{2} (x_{1} \otimes_{3} x_{2}))) \otimes_{4} M$, where we have shown labels for $x$. Substituting $M$ on the first occurrence of $x$, the event will be labelled $(4, 1)$. If we reduce the inner redex by performing events $(2,1)$ and $(2,2)$, we get
\begin{equation*}
    (\lambda x. ((x_{(1,(2,1))} \otimes_{(3,(2,1))} x_{(2,(2,1))}) \otimes_{1} (x_{(1,(2,2))} \otimes_{(3,(2,2))} x_{(2,(2,2))})) ) \otimes_{4} M
`\end{equation*}
Substituting $M$ on the first occurrence of $x$, the event will be labelled $(4, (1,(2,1)))$. If we didn't update the labels of $x$, this event will be labelled $(4,1)$, same as before. With this in mind, we now define the set of all $\lambda-$terms, denoted $\Tau(\Lambda_{E}, \Lambda_{V})$, with the set of event labels $\Lambda_{E}$, labels for variables $\Lambda_{V}$. Let $E_{L} : \Tau(\Lambda_{E}, \Lambda_{V}) \rightarrow \text{Fin}(\mathcal{E})$ (where $\mathcal{E} = \Lambda_{E} \times (\Lambda_{E} \times \Lambda_{V})^{*}$) denote the set of labels for events in a $\lambda-$expression. Let \textbf{Var} be the set of variables, and let $V_{L} : \Tau(\Lambda_{E}, \Lambda_{V}) \times \textbf{Var} \rightarrow \text{Fin}(\mathcal{V})$ (where $\mathcal{V} = \Lambda_{V} \times (\Lambda_{E} \times \Lambda_{V})^{*}$) denote the labels assigned to a particular variable in a $\lambda$-expression.
\begin{defi}[$\Tau(\Lambda_{E}, \Lambda_{V})$] 
    If $x \in \textbf{Var}$ and $a \in \mathcal{V}$, then $x_{a} \in \Tau(\Lambda_{E}, \Lambda_{V})$,  $E_{L}(x_{a}) = \emptyset$, and $V_{L}(x, x) = a$, $V_{L}(x, b)= \empty$ for $b \neq a$. If $M,N \in \Tau(\Lambda_{E}, \Lambda_{V})$ such that $M,N \in \Tau(\Lambda_{E}, \Lambda_{V})$, $E_L(M) \cap E_L(N) = V_{L}(M,x) \cap V_{L}(N,x) = \emptyset$ for all $x \in \textbf{Var}$, and $l \in \mathcal{E} \setminus (E_{L}(M) \cup E_{L}(N))$ then $M \otimes_{l} N \in \Tau(\Lambda_{E}, \Lambda_{V})$ with $E_{L}(M \otimes_{l} N) = E_{L}(M) \cup E_{L}(N) \cup \{l\}$, $V_{L}(M \otimes_{e} N, x) = V_{L}(M,x) \cup E_{L}(N,x)$. If $M \in  \Tau(\Lambda_{E}, \Lambda_{V})$, then $(\lambda x. M, n) \in \Tau(\Lambda_{E}, \Lambda_{V})$ where $n$ is the number of free occurrences of $x$ in $M$, then $V_{L}((\lambda x. M, n), x) = \emptyset$ and $V_{L}((\lambda x. M,n), y) = V_{L}(M,y)$ for $y \neq x$. 
\end{defi}
\noindent Let $E : \mathcal{E} \rightarrow \Lambda_{E}$, $H_{E} : \mathcal{E} \rightarrow (\Lambda_{E} \times \Lambda_{V})^{*}$ be the canonical projection maps. $H_{E}(e)$ is the \textit{history} of event $e$, motivated by the example presented above. The string of characters from $(\Lambda_{E} \times \Lambda_{V})^{*}$ is appended each time an event occurs. We define maps $V, H_{V}$ similarly. Let $A : \Tau(\Lambda_{E}, \Lambda_{V}) \times (\Lambda_{E} \times \Lambda_{V}) \rightarrow \Tau(\Lambda_{E}, \Lambda_{V})$ be a map that takes an expression $M$ and $(l,m)$, appending $(l,m)$ to the history of each variable and event label in $M$. Also, if $m \in V_{L}(M,x)$, then denote $M[N/x, m]$ denote the substitution of $N$ in the place of $x_{m}$.  We are now in a position to define the relation $\rightarrow_{\beta}$:
\begin{center}
\begin{tabular}{cc}
  \begin{bprooftree}
    \AxiomC{$M \rightarrow_{\beta} M'$}
    \UnaryInfC{$(\lambda  x. M, n )\rightarrow_{\beta} (\lambda x. M', n')$} 
  \end{bprooftree} &
  \begin{bprooftree}
      \AxiomC{$n \geq 2$}
      \AxiomC{$m \in V_{L}(M,x)$}
     \BinaryInfC{$(\lambda x.M,n) \otimes_{l} N \rightarrow_{\beta} (\lambda x.(M[A(N,(l,m))/x,m]),n-1) \otimes_{l} N$}
  \end{bprooftree} \\[2em]
     
  \begin{bprooftree}
      \AxiomC{$n = 1$}
      \AxiomC{$m \in V_{L}(M,x)$}
      \BinaryInfC{$(\lambda x.M,n) \otimes_{l} N \rightarrow_{\beta} M[A(N,(l,m))/x,m]$}
  \end{bprooftree}&
   \begin{bprooftree}
     \AxiomC{}
     \UnaryInfC{$(\lambda x.M,0) \otimes_{l} N \rightarrow_{\beta} M$}
  \end{bprooftree} \\[2em]
  \begin{bprooftree}
      \AxiomC{$M \rightarrow_{\beta} M'$}
      \UnaryInfC{$M \otimes_{l} N \rightarrow_{\beta} M' \otimes_{l} N$} 
  \end{bprooftree}&
  \begin{bprooftree}
      \AxiomC{$N \rightarrow_{\beta} N'$}
      \UnaryInfC{$M \otimes_{l} N \rightarrow_{\beta} M \otimes_{l} N'$} 
  \end{bprooftree}
  \end{tabular}
\end{center}
\noindent The second production rule shows the one-step substitution we explained before.
\begin{defi}[Event]
    If $A \rightarrow_{\beta} B$, then by the production rules, we know that the reduction is a substitution by the event labeled $l$ at the variable named $m$. The event of this transition is the tuple $(A,B,l,m) \in \Tau(\Lambda_{E}, \Lambda_{V}) \times \Tau(\Lambda_{E}, \Lambda_{V}) \times \Lambda_{E} \times \Lambda_{V}$. We can label the transition $A \rightarrow_{\beta, (l,m)} B$.
\end{defi}
\begin{notation}
    Assume $e_{1} = A \rightarrow_{\beta, (l,m)} B$, $e_{2} = B \rightarrow_{\beta,(l',m')} C$ be 2 events. If $l' \in E_{L}(A)$, let $\pi(l')= l'$. Otherwise, $(\exists M,N)$ and $l''$ such that $l'' \in E_{L}(N)$ and $e_{1}$ is the reduction of the subexpression $(\lambda x. M) \otimes_{l} N$ in $A$. Thus, $l' = (l'', (l,m))$. In that case, $\pi(l') = l''$. 
\end{notation}
\begin{defi}[$\leq_{\mathcal{C},2}$]
    Let $e_{1} = A \rightarrow_{\beta, (l,m)} B$, $e_{2} = B \rightarrow_{\beta,(l',m')} C$. Then, $e_{1} \leq_{\mathcal{C},2} e_{2}$ if in $A$, the application with label $\pi(l')$ is of the form $\text{expr}_{1} \otimes_{\pi(l')} \text{expr}_{2}$ where $\text{expr}_{1}$ is not a $\lambda-$abstraction. 
\end{defi}
\noindent Now that we know whether any 2 successive events are causally related, can we find a suitable notion of causality for any 2 events along a path? In the good case, we used the diamond property of the reduction relation to prove that 2 successive events are not causally related iff  their order of occurrence can be interchanged. However, it is easy to see that $\rightarrow_{\beta}$ on all $\lambda-$expressions does not satisfy the diamond property. For instance, if $n \geq 2$, then the transitions $(\lambda x. M,n)\otimes_{l} N \rightarrow_{\beta} (\lambda x. M,n)\otimes_{l} N'$ and $(\lambda x. M,n)\otimes_{l} N \rightarrow_{\beta} (\lambda x. M[N/x,m],n-1) \otimes_{l} N$ do not merge again in one step. We have $(\lambda x. M,n)\otimes_{l} N' \rightarrow_{\beta} (\lambda x. M[N'/x,m],n-1)\otimes_{l} N'$ and $(\lambda x. M[N/x,m],n-1) \otimes_{l}  N \rightarrow_{\beta} (\lambda x. M[N/x,m],n-1) \otimes_{l}  N' \rightarrow_{\beta} (\lambda x. M[N'/x,m],n-1) \otimes_{l}  N'$. So an extra step is needed. 
\newline
\\
\noindent If we have a causal structure $\mathcal{C}$ on $(\Tau(\Lambda_{E}, \Lambda_{V}), \rightarrow_{\beta})$, we can induce a causal structure on the ordinary $\lambda-$calculus. The ordinary $\lambda-$calculus is  $((\Tau(\Lambda_{E}, \Lambda_{V}),\rightarrow_{\Tilde{\beta}})$ where the production rules are:
\begin{center}
\begin{tabular}{cc}
  \begin{bprooftree}
    \AxiomC{$M \rightarrow_{\Tilde{\beta}} M'$}
    \UnaryInfC{$(\lambda  x. M, n )\rightarrow_{\Tilde{\beta}} (\lambda x. M', n')$} 
  \end{bprooftree} &
   \begin{bprooftree}
     \AxiomC{}
     \UnaryInfC{$(\lambda x.M,n) \otimes_{l} N \rightarrow_{\Tilde{\beta}} M[N/x]$}
  \end{bprooftree} \\[2em]
  \begin{bprooftree}
      \AxiomC{$M \rightarrow_{\Tilde{\beta}} M'$}
      \UnaryInfC{$M \otimes_{l} N \rightarrow_{\Tilde{\beta}} M' \otimes_{l} N$} 
  \end{bprooftree}&
  \begin{bprooftree}
      \AxiomC{$N \rightarrow_{\Tilde{\beta}} N'$}
      \UnaryInfC{$M \otimes_{l} N \rightarrow_{\Tilde{\beta}} M \otimes_{l} N'$} 
  \end{bprooftree}
  \end{tabular}
\end{center}
\noindent where $M[N/x]$ denotes the substitution of $A(N, (l,m))$ at the place of $x_m$ for all $m \in V_{L}(M,x)$. It is clear that we have $\rightarrow_{\Tilde{\beta}}^{*} \subset \rightarrow_{\beta}^{*}$. We have the following construction:
\begin{defi}[Induced causal relation]
    Let $\rightarrow_{1}, \rightarrow_{2}$ be 2 relations of a set where $\rightarrow_{2}^{*} \subset \rightarrow_{1}^{*}$. Then, we can induce a causal structure on $\mathcal{P}(\rightarrow_{2})$ as follows. Let $A_{0} \overset{e_{1}}{\rightarrow_{2}} A_{1} \rightarrow_{2} ... \rightarrow_{2}  A_{n-1} \overset{e_{n}}{\rightarrow_{2}} A_{n}$ be a path. Then, $e_{1} \leq_{\mathcal{C}} e_{n}$ if for all paths $p$ between $A_{0}$ and $A_{n}$ in $\rightarrow_{1}$, where $p = A_{00} \overset{e_{11}}\rightarrow_{1} {A_{01}}  ... \rightarrow_{1} A_{0 f(0)} \overset{e_{1f(1)}}\rightarrow_{1} A_{1} \rightarrow_{1} ... \rightarrow_{1} A_{n-1} \overset{e_{n1}}\rightarrow_{1} A_{(n-1)1} \rightarrow_{1} ... \rightarrow_{1} A_{(n-1)f(n-1)} \overset{e
_{nf(n)}}{\rightarrow_{1}} A_{n}$, there exist a $j \in \{1, ..., f(1)\}$ such that $(\forall k \in [1, f(n)]) e_{1j}\leq_{\mathcal{C}} e_{nk}$
\end{defi}
\noindent It can be seen that the induced causal structure on $(\Tau(\Lambda_{E},\Lambda_{V}), \rightarrow_{\Tilde{\beta}})$ from a causal structure on $(\Tau(\Lambda_{E},\Lambda_{V}), \rightarrow_{\beta})$ is the one discussed at the beginning of this section. However, an algorithm to determine causal relationships between 2 arbitrary events on a path $p \in \mathcal{P}(\rightarrow_{\beta})$ is not known to the author. Though it is suspected that this problem is not undecidable.
\section{Future work}

One of the key open problems is to develop an algorithm that can determine causal relationships within the full $\lambda$-calculus or to establish its undecidability. This problem likely remains decidable, as it pertains to the syntactic properties of $\lambda$-calculus rather than its semantics. Another potential research direction involves extending the work on hypergraph rewriting to identify causal relations between non-successive events. A similar approach to the one outlined in Section 4 could be applicable here. Specifically, a labeling function might be formalized as a functor $F: \mathcal{H} \rightarrow \textbf{Set}$, where the morphisms in $\mathcal{H}$ represent graph productions derived from $\rightarrow_{\beta}$. Given that the definition of causality between successive events holds for rewriting in any adhesive category, it may not be far-fetched to propose a more general theory of causal structures on paths purely in categorical terms.
\printbibliography

\end{document}